\documentclass[a4, 11pt, reqno, final]{amsart}

\usepackage{amsmath, amscd, amssymb, amsthm, amsfonts}
\usepackage{mathtools, mathrsfs}
\usepackage{graphics, graphicx}
\usepackage[colorlinks=true, urlcolor=blue, citecolor=blue, linkcolor=blue, raiselinks=true, breaklinks=true]{hyperref}
\usepackage{enumerate}
\usepackage{color}
\usepackage[numbers, sort&compress]{natbib}
\usepackage{placeins, caption, subcaption}
\usepackage{tabularx}
\usepackage{arydshln}
\usepackage[colorinlistoftodos]{todonotes}
\usepackage[notref,notcite]{showkeys}
\usepackage[margin=1in]{geometry}
\usepackage[utf8]{inputenc}
\usepackage{tikz, pgfplots, tikz-cd}
\usepackage{adjustbox}

\usepackage{bm}
\usepackage{bbm}

\def\R{\mathbb{R}}

\def\tr{\mathop {\rm tr}\nolimits}

\newcommand{\bfA}{{\bf A}}
\newcommand{\bfB}{{\bf B}}
\newcommand{\bfC}{{\bf C}}
\newcommand{\bfD}{{\bf D}}
\newcommand{\bfE}{{\bf E}}
\newcommand{\bfF}{{\bf F}}
\newcommand{\bfG}{{\bf G}}

\newcommand{\bfI}{{\bf I}}

\newcommand{\bfK}{{\bf K}}
\newcommand{\bfL}{{\bf L}}
\newcommand{\bfM}{{\bf M}}
\newcommand{\bfN}{{\bf N}}

\newcommand{\bfP}{{\bf P}}

\newcommand{\bfT}{{\bf T}}
\newcommand{\bfU}{{\bf U}}
\newcommand{\bfV}{{\bf V}}

\newcommand{\bfa}{{\bf a}}

\newcommand{\bfc}{{\bf c}}

\newcommand{\bfe}{{\bf e}}
\newcommand{\bff}{{\bf f}}
\newcommand{\bfg}{{\bf g}}

\newcommand{\bfp}{{\bf p}}

\newcommand{\bfu}{{\bf u}}
\newcommand{\bfv}{{\bf v}}

\newcommand{\bfx}{{\bf x}}
\newcommand{\bfy}{{\bf y}}

\newcommand{\calE}{\mathcal{E}}

\newcommand{\calL}{\mathcal{L}}

\newcommand{\calS}{\mathcal{S}}

\newcommand{\calU}{\mathcal{U}}

\newcommand{\bbC}{\mathbb{C}}
\newcommand{\bbR}{\mathbb{R}}

\def\bfzero{\boldsymbol 0}

\def\gammab{\boldsymbol{\varepsilon}}
\def\gammab{\boldsymbol{\gamma}}

\def\phib{\boldsymbol{\varphi}}

\theoremstyle{plain}
\newtheorem{theorem}{Theorem}[section]
\newtheorem{lemma}[theorem]{Lemma}

\newcommand{\QED}{\hspace{1ex}\hfill$\blacksquare$\vspace{2ex}}

\newcommand*{\tran}{^{\mkern-1.5mu\mathsf{T}}}

\newcommand{\inv}{^{\raisebox{.2ex}{$\scriptscriptstyle-1$}}}

\title{Volumetric Growth in Linear Elasticity Driven by
	an Optimality Criterion}

\author{Rohan Abeyaratne}
\address[R. Abeyaratne]{Department of Mechanical Engineering, Massachusetts Institute of Technology, Cambridge, MA, USA}
\email{rohan@mit.edu}

\author{Roberto Paroni}
\address[R. Paroni]{Department of Civil and Industrial Engineering, University of Pisa, Largo Lucio Lazzarino 1, 56122, Pisa, Italy}
\email{roberto.paroni@unipi.it}

\author{Marco Picchi Scardaoni}
\address[M. Picchi Scardaoni]{Department of Civil and Industrial Engineering, University of Pisa, Largo Lucio Lazzarino 1, 56122, Pisa, Italy}
\email{marco.picchiscardaoni@ing.unipi.it}

\begin{document}
	\sloppy

\title{Optimization-Driven Volumetric Growth in a Linearized Elasticity Framework}

	\date{\today}

	\begin{abstract}
Using linearized elasticity as a convenient mechanical framework, we show that volumetric growth can be formulated as an optimization-driven process in which the growth tensor is determined implicitly by constrained optimization rather than prescribed through phenomenological evolution laws. At each incremental step, the displacement and growth fields satisfy equilibrium, mass-balance constraints, and an irreversibility condition enforcing accretive growth, while an objective functional encodes the driving mechanism of the process.
 Finite element discretization leads to a finite-dimensional constrained minimization problem in the growth variables alone and makes explicit the interpretation of the evolution as a projected gradient flow. Numerical examples illustrate the proposed framework.
	\end{abstract}

	\maketitle
\tableofcontents

	\section{Introduction}

Volumetric growth describes the evolution of a body through the addition of mass generated throughout its bulk. This mechanism governs a wide range of biological and physical systems, including active media, cellular aggregates, and crystalline clusters. In such systems, newly formed material progressively integrates into the solid, thereby altering both its geometry and mechanical response over time.

The modeling of growth in elastic bodies has a long tradition in continuum mechanics, with major developments driven by applications to biological tissues, morphogenesis, tumor growth, additive manufacturing, and growth-induced instabilities. Volumetric growth is typically described by introducing internal variables, most commonly a \emph{growth tensor}, whose evolution is prescribed through suitable kinetic laws. A foundational framework was introduced in the seminal work of Skalak et al.~\cite{Skalak1982}, and the general morphoelastic theory of growth was systematically developed following the influential paper by Rodriguez et al.~\cite{Rodriguez1994}; see also the monographs by Goriely~\cite{Goriely2017} and Taber~\cite{Taber2020}.

A widely adopted approach relies on the multiplicative decomposition of the deformation gradient into elastic and growth contributions, namely $\bfF = \bfF_e \bfF_g$. This decomposition was originally introduced by Kröner and Lee in the context of finite plasticity in the 1960s (see~\cite{DelPiero2018} for a detailed discussion), where the deformation gradient is decomposed into elastic and plastic parts. In the context of volumetric growth, as proposed in~\cite{Rodriguez1994}, the plastic component is replaced by a growth tensor, which accounts for volumetric expansion and stress generation.

Within this framework, growth is typically governed by an evolution law for the internal variable, often depending on stress, strain, and/or biochemical stimuli (see, e.g., Di Carlo and Quiligotti~\cite{DiCarloQuiligotti2002}, Ambrosi and Mollica~\cite{Ambrosi2002}, Ambrosi and Guana~\cite{Ambrosi2007}, and Erlich and Zurlo~\cite{Erlich2024,Erlich2025}). Subsequent developments incorporated residual stresses and geometric incompatibility (e.g., Hoger~\cite{Hoger1986}, Chen and Hoger~\cite{Chen2000}, Epstein and Maugin~\cite{Epstein2000}), showing that incompatible growth can generate internal stresses even in the absence of external loads. These models also predict growth-induced instabilities and pattern formation (e.g., Goriely and Ben Amar~\cite{Goriely2005}, Moulton et al.~\cite{Moulton2013}).

Alternative approaches are based on differential geometry~\cite{Yavari2010,SozioYavari2019}, where the growing body is modeled as a Riemannian material manifold with an evolving metric. In this setting, internal stresses are associated with the nontrivial curvature of the manifold. It can be shown that this geometric framework is closely related to the multiplicative decomposition approach.

In contrast, the linearized elasticity regime has been comparatively less explored~\cite{jones1995,araujo2004}. In this setting, the counterpart of the multiplicative decomposition is an additive decomposition. A symmetric second-order tensor field $\bfE_g$ (referred to as growth strain or eigenstrain) is introduced and interpreted as an inelastic strain that modifies the stress-free configuration. The elastic strain is then given by $\bfE(\bfu) - \bfE_g$, where $\bfE(\bfu)$ denotes the symmetric gradient of the displacement field $\bfu$.

Such formulations are common in thermoelasticity, plasticity, and phase transformation models, where the inelastic strain is either prescribed or governed by a constitutive evolution equation. While these models capture stress redistribution due to incompatible strains, they remain essentially descriptive: the spatial distribution of growth is determined by an assumed kinetic law rather than selected through a mechanical principle.

In this work, instead of prescribing an evolution equation for $\bfE_g$, we determine its increment at each discrete time step by solving a constrained optimization problem. The unknowns are the displacement field $\bfu$ and the growth tensor field $\bfE_g$, subject to the following constraints:
\begin{enumerate}[(i)]
    \item the equilibrium equations of linear elasticity;
    \item prescribed mass increments, imposed either globally or locally;
    \item a monotonicity condition enforcing purely accretive growth.
\end{enumerate}

The objective functional encodes the driving mechanism of the process. It may represent mechanical criteria (e.g., minimization of external work) or geometric quantities (e.g., perimeter reduction). Growth thus emerges as an optimization-driven process: among all admissible growth distributions, the system selects the one that optimizes the chosen criterion while remaining mechanically admissible.

A surface growth version of this framework was proposed in~\cite{Abeyaratne2026} for one-dimensional problems. Here, we extend the theory to volumetric growth in a two-dimensional setting, with a straightforward generalization to three dimensions. A complementary formulation of optimally driven growth in nonlinear elasticity is currently in preparation.

The focus of this work is not on specific applications, but rather on the development of the underlying theoretical framework. Despite its simplicity, the model reveals several structural features of optimality-driven growth. In particular, prestress associated with accretion can affect the convexity of the compliance functional, as already observed in~\cite{Abeyaratne2026} in the one-dimensional case. When the functional is convex, the incremental problem admits a unique minimizer and the evolution is stable; when convexity is lost, nonuniqueness and localization may arise.

To prevent configurational discontinuities in time, we introduce a quadratic regularization term penalizing large deviations from the previous configuration. This approach is related to De Giorgi's theory of minimizing movements~\cite{DeGiorgi1993,Ambrosio2005}, and leads, in the time-continuous limit, to a constrained gradient flow.

The emergence of a gradient-flow structure is not entirely new. Similar evolution mechanisms arise in computational models for bulk growth, such as vertex models~\cite{fletcher2013,barton2017,PicchiScardaoni2022}, where the system evolves by minimizing a discrete energy along a (typically unconstrained) gradient descent trajectory.

After finite element discretization, the equilibrium equations reduce to a linear system, allowing elimination of the displacement field and yielding a finite-dimensional constrained minimization problem in the growth variables alone. This structure highlights the variational nature of the evolution and its interpretation as a projected gradient flow in the space of inelastic strains, regularized via a minimizing-movements scheme.

The paper is organized as follows. Section~\ref{seclgp} introduces the optimal growth problems in the continuous setting. Section~\ref{secfem} presents the corresponding finite element discretization. Section~\ref{secnum} reports the main numerical results and their interpretation. Finally, Section~\ref{secana} derives explicit, albeit approximate, solutions to the optimal evolution problems, providing further insight into the resulting growth dynamics.

	\section{The growth problem in linear elasticity}\label{seclgp}

In this section, we formulate the incremental growth problem in the setting of linearized elasticity, adopted here only as a simple framework in which the optimization-driven character of growth can be made explicit.
For clarity of exposition, we restrict attention to the two–dimensional setting; however, all arguments extend in a straightforward manner to three dimensions.

	Let $\Omega \subset \bbR^2$ be a bounded, simply connected domain with sufficiently regular boundary, representing the reference (initial, stress-free) configuration of a linearly elastic body.
	The boundary $\partial\Omega$ is decomposed into two disjoint parts:
	a Dirichlet portion $\partial_D\Omega$, where the body is clamped, and	a Neumann portion $\partial_N\Omega=\partial\Omega\setminus\partial_D\Omega$, where surface loads are prescribed.
	We introduce the space of kinematically admissible displacements\footnote{In a rigorous functional setting, one may take $\calU_0$ as the set of functions in  $H^1(\Omega, \bbR^2)$, the Sobolev space of square–integrable vector fields with square–integrable gradients, vanishing on $\partial_D\Omega$ in the trace sense. We omit further functional-analytic details, since our focus is on the mechanical aspects of the growth formulation.}
	\begin{equation}\label{U0}
		\calU_0 =
		\left\{
		\boldsymbol{\psi} : \Omega \to \mathbb{R}^2 \ \middle| \
		\boldsymbol{\psi} \text{ is regular enough and }
		\boldsymbol{\psi} = \mathbf{0} \text{ on } \partial_D \Omega
		\right\}.
	\end{equation}
	We assume for simplicity that the body undergoes only prescribed surface dead loads $\bfp:\partial_N\Omega\to\bbR^2$.

	Within the context of linearized elasticity, the displacement field $\bfu \in \calU_0$ minimizes the elastic energy
	\[\calE(\bfu) = \frac12\int_\Omega \bbC \bfE(\bfu)\cdot \bfE(\bfu) \ d\calL^2 - \int_{\partial_N\Omega} \bfp\cdot \bfu  \ d\calL^1\]
	on $\calU_0$. In the above expression, $\bbC$ is the fourth order elasticity tensor that maps symmetric tensors into symmetric tensors and satisfies the usual symmetry and positive-definiteness conditions. The tensor field $\bfE(\bfv) = \frac12(\nabla \bfv + \nabla \bfv\tran)$ is the (linearized) strain tensor associated with $\bfv \in \calU_0$, and $\bfT(\bfv) = \bbC \bfE(\bfv)$ is the corresponding Cauchy stress tensor.
	Equivalently, $\bfu$ solves the (weak) equilibrium problem
	$$\int_\Omega \bbC \bfE(\bfu)\cdot \bfE(\phib) \ d\calL^2 = \int_{\partial_N\Omega} \bfp\cdot \phib  \ d\calL^1\ \qquad \forall \phib \in\calU_0.
	$$
	It is also well-known that the trace of the strain tensor $\tr \bfE(\bfv)$ is the first order approximation of the relative local volume change due to the displacement $\bfv \in \calU_0$.

	Let $\calS$ be the vector space of symmetric tensor fields on $\Omega$:
	\begin{equation}\label{calE}
		\calS=\{\bfA :\Omega\to\bbR^{2\times 2}, \ \bfA = \bfA\tran\}.
		\end{equation}
	Given $\bfv \in \calU_0$, the strain field $\bfE(\bfv) \in \calS$.
	Nevertheless, a generic element $\bfE\in\calS$ does not necessarily coincide with the symmetric gradient of a displacement field (unless the Saint–Venant compatibility conditions are satisfied).

	To incorporate growth, we consider $\bfE_g \in \calS$ and we adopt the constitutive prescription commonly used to model anelastic effects (e.g., thermal effects, plasticity, etc.) (see for instance \cite{Taber2020, Davoli2024})
	\[\bfT(\bfv, \bfE_g) = \bbC[\bfE(\bfv)- \bfE_g].\]
	Indeed, $\bfE_g$ is interpreted as an inelastic strain that modifies the local stress-free configuration of the material.
	The quantity $\bfE(\bfv) - \bfE_g$ is the elastic strain, i.e., the part of the total strain that generates stress. Growth therefore acts essentially by shifting the stress-free state.

	Accordingly, given $\bfE_g$, the equilibrium displacement field $\bfu\in \calU_0$ verifies
	$$\int_\Omega \bbC [\bfE(\bfu) - \bfE_g]\cdot \bfE(\phib) \ d\calL^2 = \int_{\partial_N\Omega} \bfp\cdot \phib  \ d\calL^1\ \qquad \forall \phib \in\calU_0.
	$$
	Thus, growth enters in the mechanical problem as an internal source term.
	Reasoning as before, given $\bfE_g\in\calS$, the trace of the growth tensor $\tr\bfE_g$ can be interpreted as a local measure of the relative volume change due to $\bfE_g$. Consequently, the scalar $\int_\Omega \tr\bfE_g(\bfx) \ d\calL^2$ is the total change of volume (or mass, assuming constant density) of the body due to $\bfE_g$. Since  $\bfE_g$
	encodes inelastic volumetric changes, we call it a (linearized) growth tensor.

	Usually growth is a time dependent process. We are thus interested in time evolving growth tensors, at least in a discrete sense.
	The temporal evolution will be referred to as the growth process.
	Specifically, we consider a discrete evolution consisting of time steps. Quantities associated with the $i$-th step are denoted by a superscript $^{(i)}$, e.g., $\bfE_g^{(i)}$.
	For each iteration $i\geq 0$, given $\bfE_g^{(i)}$, $\bfu^{(i)} \in \calU_0$ satisfies
	$$\displaystyle\int_\Omega \bbC [\bfE(\bfu^{(i)}) - \bfE_g^{(i)}]\cdot \bfE(\phib) \ d\calL^2 = \int_{\partial_N\Omega} \bfp\cdot \phib  \ d\calL^1\ \qquad \forall \phib \in\calU_0.$$

The mass increment (assuming unit mass density) between steps $i-1$ and $i$
may be specified either globally or locally.

In the former case, for every iteration $i$, it can be specified in two ways:
\begin{equation}\label{ttg}
	\int_\Omega \displaystyle \tr(\bfE_g^{(i)})(\bfx) \ d\calL^2\ \ {=\atop \leq}\ \int_\Omega\tr(\bfE_g^{(i-1)})(\bfx) \ d\calL^2 + \Gamma^{(i)}|\Omega|
\end{equation}
where $\Gamma^{(i)}\in\bbR$ and $|\Omega|$ denotes the volume of $\Omega$.

In the case of equality, the quantity $\Gamma^{(i)}|\Omega|$ represents the volume of material added between steps $i-1$ and $i$; in the case of inequality, $\Gamma^{(i)}|\Omega|$ represents the volume available for growth.

In words, \eqref{ttg} with equality states that the global volume change in the body due to the tensor field $\bfE_g^{(i)}$ (left-hand side) is equal to the global volume change at the previous iteration plus the supplied volume of new material.

A similar interpretation holds for the inequality case. The main difference is that, with the inequality sign, the provided material may be partially or entirely unused.

When the material increment is specified locally, we consider the pointwise counterpart of \eqref{ttg}:
\begin{equation}\label{ttg2}
	\tr(\bfE_g^{(i)})(\bfx) \ \ {=\atop \leq}\ \tr(\bfE_g^{(i-1)})(\bfx) + \Gamma^{(i)}(\bfx) \qquad \forall \bfx \in \Omega
\end{equation}
where now $\Gamma^{(i)}:\Omega\to\bbR$ is a function prescribing the local mass increment (or possible increment) at iteration $i$.

Clearly, with a local prescription of the mass increment, growth is determined pointwise in $\Omega$, whereas with a global prescription it is imposed only in an average sense, so that some points may grow significantly more than others.

	We restrict our attention to purely accretive processes: material can only be added locally. This is tantamount to requiring that during growth material fibers do not shorten. That is
$$
	\bfE_g^{(i)}(\bfx)\bfv\cdot\bfv \geq \bfE_g^{(i-1)}(\bfx)\bfv\cdot\bfv \qquad \forall \bfv \in \bbR^2, \ \forall \bfx \in \Omega
	$$
	which is equivalent to requiring that $(\bfE_g^{(i)}-\bfE_g^{(i-1)})$ is a positive semidefinite tensor
	at each point of the body. Invoking the spectral theorem, it is useful to write this latter condition as
	\begin{equation}\label{lmin}
		\displaystyle\lambda_{min}\left((\bfE_g^{(i)}-\bfE_g^{(i-1)})(\bfx)\right) \geq 0 \qquad  \forall \bfx \in \Omega
	\end{equation}
	where $\lambda_{min}:\calS \to \bbR$ returns the smallest eigenvalue of its argument.

	In our approach, growth is not prescribed through an evolution equation. Instead, at each time step it is self-determined by solving a constrained minimization problem.
	We postulate that during the growth process an objective function is minimized at each $i$-th step. We consider functionals of the form
	\begin{equation}\label{obj}
		\Phi\left(\bfu^{(i)},\bfE_g^{(i)}, \ldots \right) + \frac{1}{2\tau} \int_\Omega |\bfE_g^{(i)}-\bfE_g^{(i-1)}|^2 \ d\calL^2
	\end{equation}
	where $\Phi$ is some scalar-valued function and $\tau>0$. We remark that the tensor field $\bfE_g^{(i-1)}$ is known at iteration $i$. The second term in \eqref{obj} plays the role of a minimizing–movements regularization.
	Roughly speaking, it is a weak notion of time derivative for $\bfE_g$ that penalizes abrupt temporal changes of growth. It is useful as a regularization term  when, for instance, $\Phi$ is not strictly convex with respect to its arguments, as seen in the examples below. For further motivation of the second term, we refer to \cite{Abeyaratne2026}, where a similar term was introduced in the context of surface growth.

The prototypical problem we deal with in this paper is the following:

\noindent
Given the initial growth $\bfE_g^{(0)}$, the regularization parameter $\frac{1}{2\tau}$, the surface load $\bfp$, the elasticity tensor $\bbC$, and the prescribed mass increments $\Gamma^{(i)}$, for all $i=1,\dots,N_{\mathrm{iter}}$, solve
\begin{equation}\label{peropt}
		\begin{cases}
			\displaystyle \min_{\bfu^{(i)}\in\calU_0\atop \bfE_g^{(i)}\in\calS} & \displaystyle \Phi\left(\bfu^{(i)},\bfE_g^{(i)}, \ldots \right) + \frac{1}{2\tau} \int_\Omega |\bfE_g^{(i)}-\bfE_g^{(i-1)}|^2 \ d\calL^2,\\
			& \displaystyle\int_\Omega \bbC [\bfE(\bfu^{(i)}) - \bfE_g^{(i)}]\cdot \bfE(\phib) \ d\calL^2 = \int_{\partial_N\Omega} \bfp\cdot \phib  \ d\calL^1\ \qquad \forall \phib \in\calU_0,\\
			& \displaystyle \text{either } \eqref{ttg} \text{ or } \eqref{ttg2}, \text{ and }\\
			& \displaystyle\lambda_{min}(\bfE_g^{(i)}-\bfE_g^{(i-1)})(\bfx) \geq 0 \qquad  \forall \bfx \in \Omega.
		\end{cases}
	\end{equation}

\subsection{A formal continuous limit}

In this section, we formally derive the continuous-in-time limit of the discrete-in-time problem \eqref{peropt}.

Let $\bfu^{(i)}[\bfE_g^{(i)}]\in\mathcal{U}_0$ be the solution of
$$
\int_\Omega \bbC [\bfE(\bfu^{(i)}) - \bfE_g^{(i)}]\cdot \bfE(\phib) \ d\calL^2 = \int_{\partial_N\Omega} \bfp\cdot \phib  \ d\calL^1\ \qquad \forall \phib \in\mathcal{U}_0,
$$
and define
$$
\widehat  \Phi\left(\bfE_g^{(i)}\right):= \Phi\left(\bfu^{(i)}[\bfE_g^{(i)}],\bfE_g^{(i)}\right).
$$

Assuming, for instance, that \eqref{ttg} holds with inequality, problem \eqref{peropt} can be written as
\begin{equation}\label{peropt2}
	\begin{cases}
		\displaystyle \min_{\bfE_g^{(i)}\in\mathcal{E}} &  \displaystyle \widehat\Phi\left(\bfE_g^{(i)}\right) + \frac 1{2\tau} \int_\Omega |\bfE_g^{(i)}-\bfE_g^{(i-1)}|^2 \ d\calL^2\\
		&  \displaystyle\int_\Omega \displaystyle \tr(\bfE_g^{(i)})(\bfx) \ d\calL^2\leq \int_\Omega\tr(\bfE_g^{(i-1)})(\bfx) \ d\calL^2 + \Gamma^{(i)}|\Omega|,\\
		& \displaystyle\lambda_{min}(\bfE_g^{(i)}-\bfE_g^{(i-1)})(\bfx) \geq 0 \qquad  \forall \bfx \in \Omega.
	\end{cases}
\end{equation}
If $\bfE_g^{(i)}$ solves this problem, then it also satisfies the following optimality conditions:
\begin{equation}\label{peropt2}
	\begin{cases}
		&\displaystyle  d\widehat\Phi\left(\bfE_g^{(i)}\right)[\bfG^{(i)}] + \int_\Omega \frac {\bfE_g^{(i)}-\bfE_g^{(i-1)}}{\tau} \cdot \bfG^{(i)} \ d\calL^2 =0 \qquad\forall \bfG^{(i)}\in\mathcal{G}^{(i)},\\
		&  \displaystyle\int_\Omega \displaystyle \tr(\bfE_g^{(i)}-\bfE_g^{(i-1)})(\bfx) \ d\calL^2\leq \Gamma^{(i)}|\Omega|,\\
		& \displaystyle\big(\bfE_g^{(i)}(\bfx)-\bfE_g^{(i-1)}(\bfx)\big)\bfv\cdot\bfv \geq 0 \qquad  \forall \bfx \in \Omega,\ \forall \bfv\in \R^2,
	\end{cases}
\end{equation}
where
\[\mathcal{G}^{(i)}:=\begin{multlined}[t]
\{\bfG:\Omega\to \bbR^{2\times 2}_{sym}, \ \int_{\Omega}\tr\bfG \ d\calL^2 = 0,  \bfG=\mathbf{0} \text{ on } \{\bfx\in\Omega \ : \ \bfE_g^{(i)}(\bfx)=\bfE_g^{(i-1)}(\bfx)\}\}
\end{multlined}\]
and where $d\widehat\Phi(\bfE_g^{(i)})$ denotes the Fréchet derivative of $\widehat\Phi$ evaluated at $\bfE_g^{(i)}$.

Let $t_i=i\tau$, for $i=0,1,2,\ldots$, and denote by $M(t)$ the total mass at time $t$. Then
$$
\Gamma^{(i)}=M(t_i)-M(t_{i-1}).
$$

Define the piecewise constant interpolants $\bfE_g^\tau$ and $\bfG^\tau$ by
$$
\bfE_g^\tau(t)=\bfE_g^{(i)},\quad \bfG^\tau(t)=\bfG^{(i)},\qquad \mbox{for }t\in [t_{i},t_{i+1}).
$$

Assuming that $\bfE_g^\tau \to \bfE_g$ and $\bfG^\tau \to \bfG$ as $\tau\to 0$, the formal limit of \eqref{peropt2} reads
\begin{equation}\label{peropt3}
	\begin{cases}
		&\displaystyle  d\widehat\Phi\left(\bfE_g\right)[\bfG] + \int_\Omega \frac {d\bfE_g}{dt} \cdot \bfG \ d\calL^2 =0 \qquad\forall \bfG\in\mathcal{G},\\
		&  \displaystyle\int_\Omega \displaystyle \tr\big(\frac {d\bfE_g}{dt}(\bfx)\big) \ d\calL^2\leq\frac {dM}{dt}|\Omega|,\\[8pt]
		& \displaystyle\big(\frac {d\bfE_g}{dt}(\bfx)\big)\bfv\cdot\bfv \geq 0 \qquad  \forall \bfx \in \Omega,\ \forall \bfv\in \R^2
	\end{cases}
\end{equation}
where
\[\mathcal{G}:=
	\{\bfG:\Omega\times(0,\infty)\to \bbR^{2\times 2}_{sym}, \ \int_{\Omega}\tr\bfG(\bfx,t) \ d\calL^2 = 0,  \bfG=\mathbf{0} \text{ on } \{(\bfx,t) \ : \ \frac{d\bfE_g}{dt}(\bfx,t)=\mathbf{0}\}\}.
\]
Finally, we observe that $d\widehat{\Phi}$ can be expressed in terms of $\Phi$ as
$$
d\widehat\Phi\left(\bfE_g\right)[\bfG] = d_1\Phi\left(\bfu[\bfE_g],\bfE_g\right)[d\bfu_\bfG] + d_2\Phi\left(\bfu[\bfE_g],\bfE_g\right)[\bfG],
$$
where $d_1\Phi$ and $d_2\Phi$ denote the Fréchet derivatives of $\Phi$ with respect to its first and second arguments, respectively, and $d\bfu_\bfG$ is the solution of
$$
\int_\Omega \bbC [\bfE(d\bfu_\bfG)]\cdot \bfE(\phib) \ d\calL^2 = \int_{\Omega} \bbC \bfG\cdot \bfE(\phib)  \ d\calL^2\ \qquad \forall \phib \in\mathcal{U}_0.
$$
As emphasized in the introduction, the proposed theory of growth does not rely on a prescribed evolution equation for $\bfE_g$; instead, the evolution arises intrinsically from the formulation. This feature is evident in problem~\eqref{peropt3}.

\subsection{Three representative problems}
	In what follows we focus on three representative cases to illustrate the predictive capabilities of the optimization-driven growth framework developed in this section:
\begin{enumerate}[(i)]
	\item A doubly-clamped beam under external work minimization;
	\item A cantilever beam under external work minimization;
	\item A free body evolving under perimeter minimization.
\end{enumerate}
In all cases, we consider a linearly isotropic elastic body with Young's modulus $E$ and Poisson's ratio $\nu$.
The initial natural (i.e., stress-free) configuration occupies the region $\Omega = (0,\ell)\times(0,h)$.

In the beam cases, the body is clamped on the segment $\partial_D\Omega = \{0\}\times(0,h)$ in the cantilever case and on $\partial_D\Omega = \{0,\ell\}\times(0,h)$ in the doubly clamped case. It is subjected to a constant distributed load $\bfp = -p\,\bfe_2$ ($p>0$) applied on the upper boundary $(0,\ell)\times\{h\}$. The remaining part of the Neumann boundary is traction-free.

In the perimeter-driven case, the boundary is entirely of Neumann type and traction-free. To eliminate rigid-body motions, two boundary points are constrained in an isostatic fashion.

We assume for simplicity, for all cases,
\begin{itemize}
	\item  equality in \eqref{ttg} and \eqref{ttg2},
	\item constant mass increment: $\Gamma^{(i)} = \Gamma = \mathrm{constant }>0$ in \eqref{ttg} and  \eqref{ttg2}  at every iteration,
	\item  zero initial growth: $\bfE_g^{(0)}=\mathbf{0}$.
\end{itemize}

The problem formulations are therefore:
	\begin{equation}\label{workopt22cont}
			\displaystyle
			\begin{cases}
					\min_{\bfu^{(i)}\in\calU_0\atop \bfE_g^{(i)}\in\calS} & \displaystyle \int_{\partial_N\Omega} \bfp \cdot \bfu^{(i)} \ d\calL^1 + \frac{1}{2\tau} \int_\Omega |\bfE_g^{(i)}-\bfE_g^{(i-1)}|^2 \ d\calL^2\\
					& \displaystyle\int_\Omega \bbC [\bfE(\bfu^{(i)}) - \bfE_g^{(i)}]\cdot \bfE(\phib) \ d\calL^2 = \int_{\partial_N\Omega} \bfp\cdot \phib  \ d\calL^1\ \qquad \forall \phib \in\calU_0,\\
					& \displaystyle\int_\Omega  \tr(\bfE_g^{(i)})(\bfx) \ d\calL^2= \int_\Omega\tr(\bfE_g^{(i-1)})(\bfx) \ d\calL^2 + \Gamma|\Omega|,\\
					& \displaystyle\lambda_{min}(\bfE_g^{(i)}-\bfE_g^{(i-1)}) \geq 0 \qquad  \text{ in } \Omega
				\end{cases}
		\end{equation}
for the beams and
	\begin{equation}\label{peropt22cont}
			\begin{cases}
					\displaystyle \min_{\bfu^{(i)}\in\calU_0\atop \bfE_g^{(i)}\in\calS} & \displaystyle P((\mathbf{id} + \bfu^{(i)})(\Omega)) + \frac{1}{2\tau} \int_\Omega |\bfE_g^{(i)}-\bfE_g^{(i-1)}|^2 \ d\calL^2\\
					& \displaystyle\int_\Omega \bbC [\bfE(\bfu^{(i)}) - \bfE_g^{(i)}]\cdot \bfE(\phib) \ d\calL^2 = 0 \qquad \forall \phib \in\calU_0,\\
					& \displaystyle\tr(\bfE_g^{(i)})(\bfx) = \tr(\bfE_g^{(i-1)})(\bfx)  + \Gamma \qquad \forall \bfx \in \Omega,\\
					& \displaystyle\lambda_{min}(\bfE_g^{(i)}-\bfE_g^{(i-1)}) \geq 0 \qquad  \text{ in } \Omega.
				\end{cases}
		\end{equation}
for the perimeter case. In the above problem the objective function is the perimeter of the deformed configuration $(\mathbf{id} + \bfu^{(i)})(\Omega)$:
\[P((\mathbf{id} + \bfu^{(i)})(\Omega)).\]
In the beam problems the mass constraint is global, whereas in the perimeter case it is local. This distinction is not essential, but serves to demonstrate the flexibility of our framework.
The spaces $\calU_0, \calS$ are defined in \eqref{U0} and \eqref{calE}.

	\section{Finite dimensional approximation of the problem}\label{secfem}

	In this section we derive a finite dimensional approximation of the incremental growth problem introduced at the end of Section \ref{seclgp} by means of the finite element method. The purpose is twofold:
	to transform the continuous constrained minimization problem into an explicit algebraic system,
	and to make transparent the structural coupling between displacement and growth variables.
	This reduction plays a crucial role in the analytical developments of Section \ref{secana}.

	We begin by rewriting the strain tensor in vector form, which is convenient for finite element implementation.
	Let $\bfe_i$ denote the  canonical basis of $\mathbb{R}^2$ and, with a slight abuse of notation, also of $\R^3$.

	For a displacement $\bfu : \Omega \to \mathbb{R}^2$ we set $\boldsymbol{\varepsilon}(\bfu):\Omega \to \bbR^3$ as
	\[
	\begin{aligned}
		\boldsymbol{\varepsilon}(\bfu) &=E_{11}(\bfu)\bfe_1+E_{22}(\bfu)\bfe_2+2E_{12}(\bfu)\bfe_3\\
		&=\dfrac{\partial u_1}{\partial x_1}\bfe_1+\dfrac{\partial u_2}{\partial x_2}\bfe_2+\big(\dfrac{\partial u_1}{\partial x_2} + \dfrac{\partial u_2}{\partial x_1}\big)\bfe_3.
	\end{aligned}
	\]
	as the vector representation of the strain $\bfE(\bfu)$.
	With this notation, the elasticity tensor is represented by a symmetric $3\times 3$ matrix $\bfC$ such that
	$$
	\bfC\boldsymbol{\varepsilon}(\bfu)\cdot \boldsymbol{\varepsilon}(\bfv)=
	\bbC \bfE(\bfu) \cdot \bfE(\bfv)
	$$
	for all displacements $\bfu$ and $\bfv$.
	For an isotropic material in plane stress conditions,
	\[
	\mathbf{C} = \frac{E}{1 - \nu^2}
	\begin{pmatrix}
		1 & \nu & 0\\[4pt]
		\nu & 1 & 0\\[4pt]
		0 & 0 & \frac12(1 - \nu)
	\end{pmatrix},
	\]
	where $E$ and $\nu$ denote the Young's modulus and the Poisson ratio,
	respectively.
	Similarly, the growth tensor is represented in vector form as
	\[
	\gammab^{(i)} =(\bfE_g^{(i)})_{11}\bfe_1+(\bfE_g^{(i)})_{22}\bfe_2+2(\bfE_g^{(i)})_{12}\bfe_3.
	\]
	With this notation, the equilibrium equation \eqref{peropt} at step $i$  rewrites as
	\begin{equation}\label{dper1}
		\int_\Omega \mathbf{C}\boldsymbol{\varepsilon}(\bfu^{(i)}) \cdot
		\boldsymbol{\varepsilon}(\phib) \, d\calL^2
		=
		\int_\Omega \mathbf{C}\gammab^{(i)} \cdot
		\boldsymbol{\varepsilon}(\phib) \, d\calL^2 + \int_{\partial_N\Omega} \bfp \cdot \phib \ d\calL^1 \ \qquad \forall \phib \in\calU_0.
	\end{equation}

	Assume for simplicity that $\Omega \subset \mathbb{R}^2$ is a polygonal domain in the plane.
	Given $h>0$, we define a \emph{triangulation}  of $\Omega$, that is, a partition of $\Omega$ into a set $\mathcal{T}_h$ of triangles $T_e$ satisfying the following properties:
	\begin{itemize}
		\item $\mathcal{T}_h$ consists of $N$ nodes (the triangle vertices), denoted by $\mathbf{x}_i$, $i=1,\ldots,N$, and $N_e$ triangles (elements), denoted by $T_e$, $e=1,\ldots,N_e$;
		\item the domain is the union of the triangles:
		$\Omega = \bigcup_{e=1}^{N_e} T_e$;
		\item the intersection of the interiors of any two distinct triangles is empty;
		\item no vertex of a triangle lies in the interior of an edge belonging to another triangle;
		\item every triangle in $\mathcal{T}_h$ has diameter less than or equal to $h$.
	\end{itemize}

	We introduce the standard nodal basis functions $ \{\eta_i\} : \Omega \to \mathbb{R}$, which are globally continuous on $\Omega$ up to the boundary, affine on each triangle of  $\mathcal{T}_h$, and satisfying the condition
	$\eta_i(\mathbf{x}_j) = \delta_{ij}$,
	$i, j = 1, \ldots, N$.\\
	We recall that for any subset $A$ of $\Omega$,  the indicator function of $A$ is the function  $1_{A}$
	defined by  $1_{A}(\bfx)=1$ if $\bfx \in A$ and $1_{A}(\bfx)=0$ otherwise in $\Omega$.
	Let
	\begin{itemize}
		\item $\widehat{\bfe}_i$, for $i = 1,2,\ldots,2N$, be the canonical basis of $\mathbb{R}^{2N}$;
		\item $\widetilde{\bfe}_i$, for $i = 1,2,\ldots,N_e$, be the canonical basis of $\mathbb{R}^{N_e}$;
		\item $\bar{\bfe}_i$, for $i = 1,2,\ldots,3N_e$, be the canonical basis of $\mathbb{R}^{3N_e}$.
	\end{itemize}
	We define the following spaces
	\[
	\begin{aligned}
		\calU^h &= \{
		\bfv : \Omega \to \mathbb{R}^2 :
		\exists\, \widehat{\bfv} \in \mathbb{R}^{2N} \ \text{such that } \bfv(\bfx) = \sum_{i=1}^N \widehat{v}_{2i-1}\eta_i(\bfx)\bfe_1+
		\widehat{v}_{2i}\eta_i(\bfx)\bfe_2
		\}\\
		\calS^h &= \{
		\gammab : \Omega \to \mathbb{R}^3 :
		\exists\, \bar{\gammab} \in \mathbb{R}^{3N_e} \ \text{such that }\\
		&\hspace{4cm}
		\gammab(\bfx) = \sum_{e=1}^{N_e} 1_{T_e}(\bfx)(\bar{\gamma}_{3e-2}\bfe_1+ \bar{\gamma}_{3e-1}\bfe_2+ \bar{\gamma}_{3e}\bfe_3)
		\}.
	\end{aligned}
	\]
	In plane language, $\calU^h$ is the space of vector-valued functions that are
	continuous on $\Omega$ and affine on each triangle of
	$\mathcal{T}_h$,
	and $\calS^h$ is the space of vector-valued functions that are
	constant on each triangle of
	$\mathcal{T}_h$ and, in general, discontinuous on $\Omega$. In other words, this choice corresponds to piecewise constant growth tensors, allowing discontinuities across element boundaries.\\
	There is a one-to-one correspondence between the functions of these spaces and the coefficient vectors that they generate:
	for instance
	$$
	\mbox{if }\bfv\in\calU^h\quad\Rightarrow\quad \exists!\ \widehat{\bfv} \in \mathbb{R}^{2N} \ \text{such that } \bfv = \sum_{i=1}^N \widehat{v}_{2i-1}\eta_i\bfe_1+
	\widehat{v}_{2i}\eta_i\bfe_2,
	$$
	and similarly for $\calS^h$.
	Moreover, we can write
	\[
	\bfv(\bfx)=\mathbf{N}(\bfx)\, \widehat{\bfv},
	\quad\mbox{and}\quad
	\boldsymbol{\varepsilon}(\bfv)(\bfx)= \mathbf{D}(\bfx)\, \widehat{\bfv}
	\]
	where the matrix of nodal basis functions $\bfN:\Omega\to\R^{2\times 2N}$ and the matrix of their symmetric gradient $\bfD:\Omega\to\R^{3\times 2N}$ are given by
	\begin{equation}\label{popt2}
		\begin{aligned}
			\mathbf{N}(\bfx) &:=\Big(
			\bfe_1 \otimes \sum_{i=1}^N \eta_i(\bfx)\, \widehat{\bfe}_{2i-1}
			+ \bfe_2 \otimes \sum_{i=1}^N \eta_i(\bfx)\, \widehat{\bfe}_{2i}
			\Big),\\
			\mathbf{D}(\bfx)
			&:=
			\sum_{i=1}^N
			\left(\frac{\partial \eta_i(\bfx)}{\partial x_1}
			\bigl(\bfe_1 \otimes \widehat{\bfe}_{2i-1}
			+ \bfe_3 \otimes \widehat{\bfe}_{2i}\bigr)
			+
			\frac{\partial \eta_i(\bfx)}{\partial x_2}
			\bigl(\bfe_2 \otimes \widehat{\bfe}_{2i}
			+ \bfe_3 \otimes \widehat{\bfe}_{2i-1}\bigr)\right).
		\end{aligned}
	\end{equation}
	To discretize problem \eqref{peropt}, we further introduce the subspace $\calU^h_0$ of $\calU^h$ of functions that satisfy the Dirichlet condition on $\partial_D\Omega$, and we consider the (vectorized) growth tensor to belong to $\calS^h$:
	\[
	\bfu^{(i)}\in\calU^h_0 :=
	\left\{
	\bfg \in \calU^h : \bfg = \mathbf{0} \text{ on } \partial_D \Omega
	\right\}
	\quad\mbox{and}\quad
	\gammab^{(i)} \in \calS^h.
	\]
	We denote by $\widehat\bfu^{(i)}\in\R^{2N}$ and $\bar\gammab^{(i)}\in\R^{3N_e}$ the corresponding coefficient vectors.
	It is clear that $\widehat\bfu^{(i)}$ belongs to
	$$
	\widehat V^h_0:=\{\widehat\bfg\in\R^{2N} \mid \widehat\bfg\cdot \widehat\bfe_{2i-1} = \widehat\bfg\cdot \widehat\bfe_{2i}=0 \ \text{ if } \bfx_i \in \partial_D\Omega \},
	$$
	the subspace of coefficient vectors relative to functions in $\calU_0^h$.

	Let $\bfP:\R^{2N}\to \widehat V^h_0$ be the orthogonal projection of $\R^{2N}$ onto $\widehat V^h_0$.
	The left side of \eqref{dper1} writes
	\begin{equation}\label{dper1l}
		\int_\Omega \mathbf{C}\boldsymbol{\varepsilon}(\bfu^{(i)}) \cdot
		\boldsymbol{\varepsilon}(\phib) \, d\calL^2
		=(\bfP\tran \int_\Omega \bfD\tran\mathbf{C}\bfD\, d\calL^2 \bfP)\ \widehat\bfu^{(i)} \cdot
		\widehat\phib
	\end{equation}
	while the right side of \eqref{dper1} leads to
	\begin{equation}\label{dper1r}
		\begin{aligned}
			\int_\Omega \mathbf{C}\gammab^{(i)} \cdot
			\boldsymbol{\varepsilon}(\phib) \, d\calL^2&=\int_\Omega \mathbf{C} \sum_{e=1}^{N_e}1_{T_e}( \bar\gamma^{(i)}_{3e-2}\bfe_1+ \bar\gammab^{(i)}_{3e-1}\bfe_2+ \bar\gammab^{(i)}_{3e}\bfe_3)\cdot\bfD\bfP\widehat\phib\, d\calL^2\\
			&=(\sum_{e=1}^{N_e}\int_{T_e} \bfP\tran\bfD\tran\mathbf{C} (\bfe_1\otimes\bar\bfe_{3e-2}+ \bfe_2\otimes\bar\bfe_{3e-1}+ \bfe_3\otimes\bar\bfe_{3e})\, d\calL^2) \bar\gammab^{(i)}\cdot\widehat\phib\\
			&=(\bfP\tran\sum_{e=1}^{N_e}(\int_{T_e} \bfD\tran\mathbf{C} \, d\calL^2) (\bfe_1\otimes\bar\bfe_{3e-2}+ \bfe_2\otimes\bar\bfe_{3e-1}+ \bfe_3\otimes\bar\bfe_{3e})) \bar\gammab^{(i)}\cdot\widehat\phib.
		\end{aligned}
	\end{equation}
	Similarly, if we denote by $\bfp^h$ the element of $\calU^h$ such that
	$\bfp^h = \bfp$ at all $\bfx_i \in \partial_N\Omega$ and $\bfp^h = \mathbf{0}$ at all $\bfx_i \in \Omega$, there exists a unique coefficient vector $\widehat{\bfp}\in \bbR^{2N}$ associated to $\bfp^h$ and such that $\bfp(\bfx) = \bfN(\bfx) \widehat{\bfp}$. Accordingly
	\begin{equation}\label{dper1rl}
		\begin{aligned}
			\int_{\partial_N\Omega} \bfp \cdot
			\phib \, d\calL^1&=(\bfP\tran\int_\Omega \bfN\tran\bfN  \ d\calL^2) \widehat{\bfp}\cdot\widehat\phib.
		\end{aligned}
	\end{equation}

	Setting
	\begin{equation}\label{dper3}
		\begin{aligned}
			{\mathbf{K}} &:=
			\bfP\tran (\int_\Omega \bfD\tran(\bfx)\mathbf{C}\bfD(\bfx)\, d\calL^2) \bfP,\\
			\bfB &:=\bfP\tran\sum_{e=1}^{N_e}(\int_{T_e} \bfD\tran(\bfx)\mathbf{C} \, d\calL^2) (\bfe_1\otimes\bar\bfe_{3e-2}+ \bfe_2\otimes\bar\bfe_{3e-1}+ \bfe_3\otimes\bar\bfe_{3e}),\\
			\bff &:= (\bfP\tran\int_\Omega \bfN\tran(\bfx)\bfN(\bfx)  \ d\calL^2) \widehat{\bfp}
		\end{aligned}
	\end{equation}
	as the stiffness matrix, the matrix coupling growth and displacements, the vector representing the external loads,
	the discrete counterpart of the equilibrium equation in \eqref{peropt} reads
	\begin{equation}\label{dper4}
		\bfK\, \widehat\bfu^{(i)}=\bfB\,\bar\gammab^{(i)} + \bff,
	\end{equation}
	which admits a unique solution in $\widehat V^h_0$.

	\subsection*{Discretization of the mass balance \eqref{ttg} and \eqref{ttg2}}
	We start with the local version of the constraint. The (discrete) trace of the growth tensor on element $T_e$ is
$\bar\gamma_{3e-2} + \bar\gamma_{3e-1}$ and is constant on the element.  Unless $\Gamma^{(i)}$ in \eqref{ttg2} is constant on each element, we may consider the pointwise constraint to hold in the mean. We can thus pose
\begin{equation}\label{pippo}
	\sum_{e=1}^{N_e} 1_{T_e}(\bfx)[(\bar\gamma^{(i)}_{3e-2} - \bar\gamma^{(i-1)}_{3e-2} + \bar\gamma^{(i)}_{3e-1} - \bar\gamma^{(i-1)}_{3e-1}] \le \sum_{e=1}^{N_e} \frac{1_{T_e}(\bfx)}{|T_e|}\int_{T_e}\Gamma^{(i)}\,d\calL^2.\end{equation}
	It is convenient to recast this inequality into a vector form. To this aim, let
	$$
	\widetilde \bfa ^{(i)}:=\sum_{e=1}^{N_e}\frac{1}{|T_e|}\int_{T_e}\Gamma^{(i)}\,d\calL^2\,\widetilde\bfe_e
	$$
	be the average value of $\Gamma^{(i)}$ on the $e$-th element,
	and let
	$$
	\bfA:=\sum_{e=1}^{N_e} \widetilde\bfe_e\otimes(\bar\bfe_{3e-2}+\bar\bfe_{3e-1}).
	$$
	Then \eqref{pippo} recasts as
	\begin{equation}\label{dper4}
		\bfA\,(\bar\gammab^{(i)}-\bar\gammab^{(i-1)})\le \widetilde \bfa ^{(i)}.
	\end{equation}
	$\bfA$ is a matrix extracting elementwise traces, and $\widetilde\bfa ^{(i)}$ contains the elementwise averages of $\Gamma^{(i)}$.

	With a similar argument, we deduce the discrete version of the global form \eqref{ttg}:
	\[\bfa\cdot(\bar\gammab^{(i)}-\bar\gammab^{(i-1)})  \leq \Gamma^{(i)}|\Omega|\]
	where
	\[\bfa := \sum_{e=1}^{N_e} |T_e| \bfA\tran\widetilde{\bfe}_e\]
	is the vector that sum the elementwise traces weighted with the corresponding element area.

	\subsection*{Discretization of growth monotonicity \eqref{lmin}}
	The positive semidefinite condition is enforced elementwise.
	It is well known that for a symmetric second order tensor $\bfA$ in any orthonormal basis
	\[\lambda_{min}(\bfA) = \frac{A_{11}+A_{22}}2-\sqrt{\big(\frac{A_{11}-A_{22}}2 \big)^2+A_{12}^2}.\]	Applying this to each element, setting
	$$
	\bfc(\bar\gammab):=-\sum_{e=1}^{N_e}\Big(\frac{\bar \gamma_{3e-2}+\bar \gamma_{3e-1}}2-\sqrt{\big(\frac{\bar \gamma_{3e-2}-\bar \gamma_{3e-1}}2 \big)^2+\frac{\bar \gamma_{3e}^2}{4}}\,
	\Big)\widetilde\bfe_e,
	$$
	yields the discrete version of \eqref{lmin}:
	\begin{equation}\label{dper5}
		\bfc(\bar\gammab^{(i)}-\bar\gammab^{(i-1)})\le \bfzero.
	\end{equation}

	\subsection*{Discretization of the objective function}
	Let us focus first on the penalization term.
	With elementary computations, we deduce
	$$
	\begin{aligned}
		\int_\Omega |\gammab^{(i)} -\gammab^{(i-1)} |^2 \ d\calL^2
		&=
		\sum_{e=1}^{N_e}\int_{T_e}|(\bfe_1\otimes\bar\bfe_{3e-2}+ \bfe_2\otimes\bar\bfe_{3e-1}+ \bfe_3\otimes\bar\bfe_{3e})(\bar\gammab^{(i)}-\bar\gammab^{(i-1)})|^2\,d\calL^2\\
		&=\frac 12 \bfL\,
		[\bar\gammab^{(i)}-\bar\gammab^{(i-1)}]\cdot(\bar\gammab^{(i)}-\bar\gammab^{(i-1)})
	\end{aligned}
	$$
	where we set
	\begin{equation}\label{eqL}
	\bfL:=2\sum_{e=1}^{N_e}|T_e|(\bar\bfe_{3e-2}\otimes\bar\bfe_{3e-2}+ \bar\bfe_{3e-1}\otimes\bar\bfe_{3e-1}+ \bar\bfe_{3e}\otimes\bar\bfe_{3e}).
	\end{equation}
	We note in passing that $\bfL$ is diagonal and positive definite, thus invertible.

	We now discretize the objective functions.

	\subsection*{Perimeter} Consider as objective function  the perimeter of the deformed configuration $(\mathbf{id} + \bfu^{(i)})(\Omega)$:
	\[\Phi(\bfu^{(i)}) = P((\mathbf{id} + \bfu^{(i)})(\Omega)).\]
	Let
	$$
	I_b:=\{i\in\{1,2,\ldots,N\}:\bfx_i\in\partial\Omega\},
	\qquad
	N_b:=|I_b|
	$$
	and $b:\{1,2,\ldots,N_b,N_b+1\}\to I_b$ a function such that $b(N_b+1)=b(1)$ and
	$$\bfx_{b(i)}\quad\mbox{and}\quad \bfx_{b(i+1)} \quad\mbox{ belong to the same element for }i=1,\ldots,N_b.
	$$
	Then, setting
	$$
	\bfy_{b(j)}:=\bfx_{b(j)}+\widehat u^{(i)}_{2b(j)-1}\bfe_1+\widehat u^{(i)}_{2b(j)}\bfe_2,
	$$
	we have that the perimeter can be expressed as
	$$
	\begin{aligned}
		P_u(\widehat\bfu^{(i)})&:= \sum_{j=1}^{N_b}|\bfy_{b(j+1)}-\bfy_{b(j)}|.
	\end{aligned}
	$$

	Let us notice that
	$$
	\begin{aligned}
		\frac{\partial P_u}{\partial \widehat u^{(i)}_{2b(j)-1}} (\widehat\bfu^{(i)})&= -\frac{\bfy_{b(j+1)}-\bfy_{b(j)}}{|\bfy_{b(j+1)}-\bfy_{b(j)}|}\cdot \bfe_1+\frac{\bfy_{b(j)}-\bfy_{b(j-1)}}{|\bfy_{b(j)}-\bfy_{b(j-1)}|}\cdot \bfe_1\\
		\frac{\partial P_u}{\partial \widehat u^{(i)}_{2b(j)}} (\widehat\bfu^{(i)})&= -\frac{\bfy_{b(j+1)}-\bfy_{b(j)}}{|\bfy_{b(j+1)}-\bfy_{b(j)}|}\cdot \bfe_2+\frac{\bfy_{b(j)}-\bfy_{b(j-1)}}{|\bfy_{b(j)}-\bfy_{b(j-1)}|}\cdot \bfe_2\\
	\end{aligned}
	$$
	and thus
	$$
	\nabla P_u=  \sum_{j=1}^{N_b}\frac{\partial P_u}{\partial \widehat u^{(i)}_{2b(j)-1}}\widehat\bfe_{2b(j)-1}+\frac{\partial P_u}{\partial \widehat u^{(i)}_{2b(j)}}\widehat\bfe_{2b(j)}.
	$$

	If $\widehat\bfu$ solves \eqref{dper4}, i.e., $\widehat\bfu = \bfK\inv(\bfB\,\bar\gammab^{(i)} + \bff)$, we define
	$$
	P_g(\bar\gammab^{(i)}):=P_u(\bfK\inv(\bfB\,\bar\gammab^{(i)} + \bff)).
	$$
	Since
	$$
	\frac{\partial \widehat\bfu^{(i)}}{\partial \bar\gamma^{(i)}_k}= \bfK\inv\bfB\,\bar\bfe_k
	$$
	we have
	$$
	\frac{\partial P_g}{\partial \bar\gamma^{(i)}_k}= \nabla P_u\cdot\bfK\inv\bfB\,\bar\bfe_k=\bfB\tran\bfK\inv\nabla P_u\cdot\bar\bfe_k
	$$
	and therefore
	$$
	\nabla P_g=\bfB\tran\bfK\inv\nabla P_u.
	$$
	The latter expression is the gradient of the perimeter with respect to $\bar\gammab^{(i)}$.

	We now show that $P_g$ is not strictly convex as a function of $\bar\gammab^{(i)}$.
	\begin{lemma}
		The map $P_g : \mathbb{R}^{3N_e} \to \mathbb{R}$,
		$\bar{\boldsymbol{\gamma}}^{(i)} \mapsto P_g\big(\bar{\boldsymbol{\gamma}}^{(i)}\big)$,
		is convex but not strictly convex.
	\end{lemma}
	\begin{proof}
		In this proof, we omit superscript $^{(i)}$ for notational convenience. For  $j=1,\dots,N_b$ we introduce the linear operators $\bfC_j:\bbR^{2N} \to \bbR^2$ defined by
		\[\bfC_j := \bfe_1\otimes \bfe_{2b(j)-1} + \bfe_2\otimes\bfe_{2b(j)}.\]
		A straightforward computation yields
		\[
		P_g(\bar\gammab) =
		\sum_{j=1}^{N_b} |		\bfx_{b(j+1)} - 		\bfx_{b(j)} +(\bfC_{j+1} - \bfC_{j})\bfK\inv\bff +(\bfC_{j+1} - \bfC_{j})\bfK\inv\bfB\bar\gammab|,
		\]
		and for notational simplicity we define $\bfc_j := \bfx_{b(j+1)} - 		\bfx_{b(j)} +(\bfC_{j+1} - \bfC_{j})\bfK\inv\bff$.
		With this notation, we first show that $P_g:\bbR^{3N_e}\to\bbR$ is convex.
		Indeed, for any pair $\bar\gammab_1,\bar\gammab_2 \in \bbR^{3N_e}$ and for all $t\in[0,1]$ we have
		\[\begin{aligned}
			P_g(t \bar\gammab_1 + (1-t)\bar\gammab_2) &= \sum_{j=1}^{N_b} |		\bfc_j +(\bfC_{j+1} - \bfC_{j})\bfK\inv\bfB(t \gammab_1 + (1-t)\gammab_2)|\\
			&= \sum_{j=1}^{N_b} |		t\bfc_j + (1-t)\bfc_j +(\bfC_{j+1} - \bfC_{j})\bfK\inv\bfB(t \gammab_1 + (1-t)\gammab_2)|\\
			&\leq t \sum_{j=1}^{N_b} |		\bfc_j  +(\bfC_{j+1} - \bfC_{j})\bfK\inv\bfB \gammab_1| + (1-t) \sum_{j=1}^{N_b} |		\bfc_j  +(\bfC_{j+1} - \bfC_{j})\bfK\inv\bfB \gammab_2|\\
			&= t P_g(\bar\gammab_1) + (1-t) P_g(\bar\gammab_2)
		\end{aligned}\]
		where the inequality follows from the triangle inequality for the Euclidean norm.
		For strict convexity, we must show that the equality sign in the above inequality holds only if $\bar\gammab_1=\bar\gammab_2$ for $t\in(0,1)$. Indeed, the injectivity of the operators $(\bfC_{j+1} - \bfC_{j})\bfK\inv\bfB $ (for all $j$) would be required:
		\[\bar\gammab_1\neq\bar\gammab_2 \implies (\bfC_{j+1} - \bfC_{j})\bfK\inv\bfB \bar\gammab_1 \neq (\bfC_{j+1} - \bfC_{j})\bfK\inv\bfB \bar\gammab_2.\]
		However, operators $(\bfC_{j+1} - \bfC_{j})\bfK\inv\bfB: \bbR^{3N_e} \to \bbR^2$ cannot be injective since $3N_e > 2$. \QED
	\end{proof}

	\subsection*{External work}
	We consider now the case
	\[\Phi = \Phi(\bfu^{(i)}) = \int_{\partial_N\Omega} \bfu^{(i)}\cdot \bfp \,d\calL^1\]
	representing the work done by the external loads. Minimizing external work corresponds to increasing structural stiffness. By reducing this quantity, the body distributes growth to oppose deformation.

	Reasoning as before, we deduce easily
	\[\Phi = \bff \cdot \widehat{\bfu}^{(i)}, \qquad \frac{\partial \Phi}{\partial \bar\gammab^{(i)}} = \bfB\tran \bfK\inv\bff. \]


	\section{Numerical experiments}\label{secnum}
	Under the finite element discretization introduced in Section \ref{secfem}, problems \eqref{workopt22cont} and \eqref{peropt22cont} read in their discrete approximation as follows:
	\begin{equation}\label{workopt22}
	\begin{cases}
		\displaystyle \min_{\widehat\bfu^{(i)}\in\widehat{V}_0^h\atop \bar\gammab^{(i)}\in\bbR^{3N_e}} & \bff\cdot\widehat\bfu^{(i)} + \frac12 \frac{1}{2\tau} \bfL (\bar\gammab^{(i)}-\bar\gammab^{(i-1)})\cdot (\bar\gammab^{(i)}-\bar\gammab^{(i-1)}),\\
		& \displaystyle \bfK\, \widehat\bfu^{(i)}=\bfB\,\bar\gammab^{(i)} + \bff,\\
		& \displaystyle \bfa\cdot(\bar\gammab^{(i)}-\bar\gammab^{(i-1)})  = \Gamma|\Omega|,\\
		& \displaystyle \bfc(\bar\gammab^{(i)}-\bar\gammab^{(i-1)})\le \bfzero
	\end{cases}
	\end{equation}
	for the beams and
	\begin{equation}\label{peropt22}
\begin{cases}
	\displaystyle \min_{\widehat\bfu^{(i)}\in\widehat{V}_0^h\atop \bar\gammab^{(i)}\in\bbR^{3N_e}} & P_u(\widehat\bfu^{(i)}) + \frac12 \frac{1}{2\tau} \bfL (\bar\gammab^{(i)}-\bar\gammab^{(i-1)})\cdot (\bar\gammab^{(i)}-\bar\gammab^{(i-1)}),\\
	& \displaystyle \bfK\, \widehat\bfu^{(i)}=\bfB\,\bar\gammab^{(i)},\\
	& \displaystyle  \bfA\,(\bar\gammab^{(i)}-\bar\gammab^{(i-1)})= \widetilde \bfa,\\
	& \displaystyle \bfc(\bar\gammab^{(i)}-\bar\gammab^{(i-1)})\le \bfzero.
\end{cases}
	\end{equation}
	for the perimeter case.

	Table \ref{table:1} summarizes the numerical values of the parameters used in the simulations.
	\begin{table}[h!]
		\centering
		\begin{tabular}{|c| c| c| c|}
			Symbol & Cantilever & Double clamped & Perimeter\\ [0.5ex]
			\hline
			$E$ & 1 & 1 & 1 \\
			$\nu$ & 0  & 0 & 0\\
			$\ell$ & 1 & 1 & 1 \\
			$h$ & 0.1 & 0.1 & 0.5 \\
			$p$ & 5e-4 & 5e-3 & - \\
			$\Gamma$ & 0.05 & 0.05 & 0.024 \\
			$\frac{1}{2\tau}$ & 10 & 10 & 100 \\
			$N$ & 306 & 306 & 715\\
			$N_e$ & 498 & 498 & 1319\\
			$N_{iter}$ & 30& 30 & 500\\
		\end{tabular}
		\caption{Numerical values of parameters involved in numerical simulations}
		\label{table:1}
	\end{table}
	The constrained optimization problems were solved using MATLAB’s \emph{fmincon} routine with the SQP algorithm. Default algorithmic parameters were used, except \emph{MaxIterations}=100, \emph{OptimalityTolerance} = \emph{StepTolerance} = 1e-4 to speed up the overall convergence.
	Analytical gradients of both the objective function and the nonlinear constraints were supplied to improve convergence.

	The mesh generation and handling relies on the open source package \emph{DistMesh} \cite{Persson2004}, while the finite element discretization has been performed through an in-house implementation.

	Some meaningful pictures have been reported in Fig. \ref{defd}, \ref{g11d}, \ref{g22d}, \ref{t11d}, \ref{t011d}, \ref{vmd}.
	\begin{figure}[!hbt]
		\centering
		\subfloat[Reference configuration]{\includegraphics[width=0.3\columnwidth]{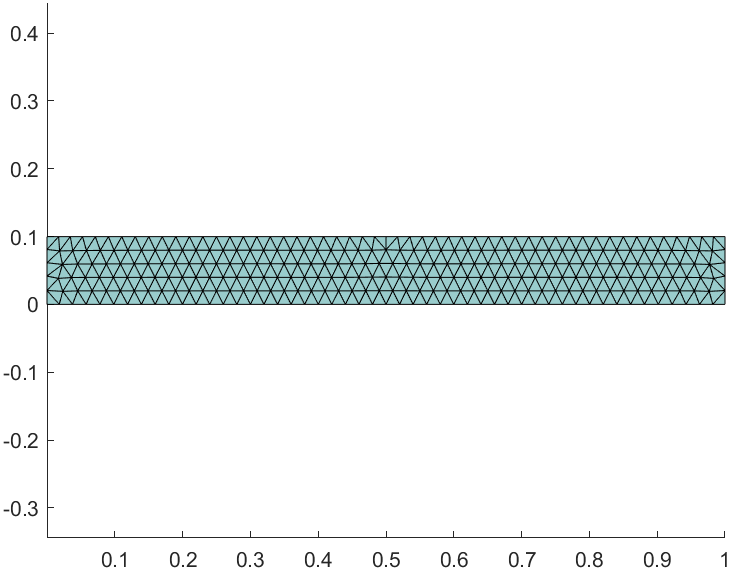}}\hfill
		\subfloat[Deformed, grown configuration after 15 timesteps]{\includegraphics[width=0.3\columnwidth]{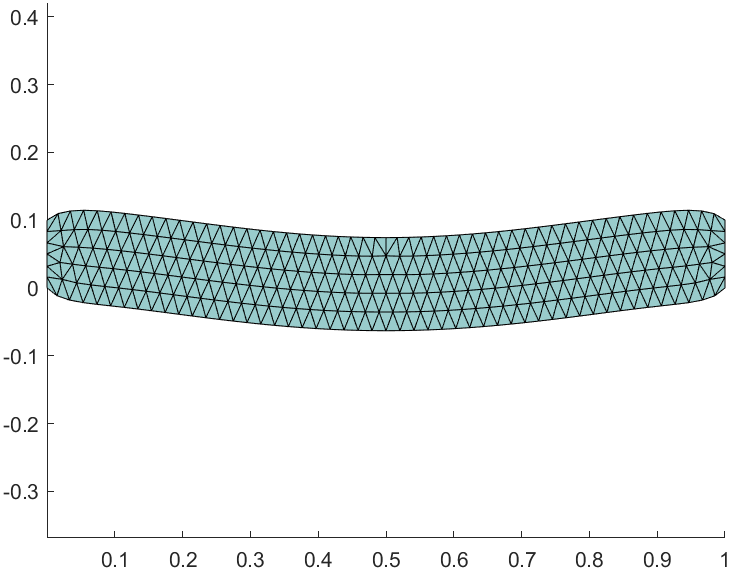}}\hfill
		\subfloat[Deformed, grown configuration after 30 timesteps]{\includegraphics[width=0.3\columnwidth]{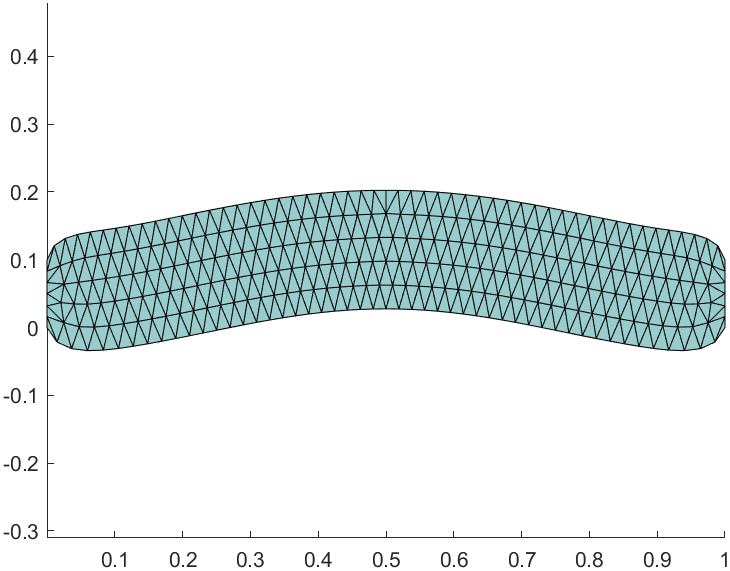}}\\
		\subfloat[Reference configuration]{\includegraphics[width=0.3\columnwidth]{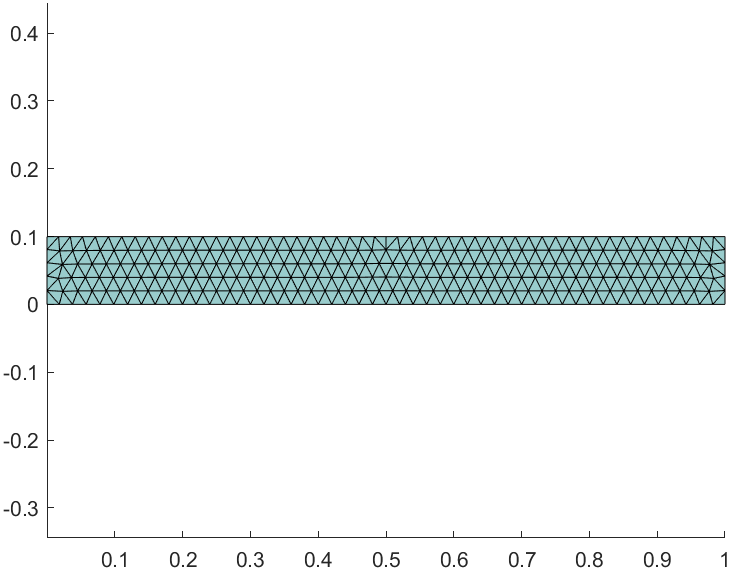}}\hfill
		\subfloat[Deformed, grown configuration after 15 timesteps]{\includegraphics[width=0.3\columnwidth]{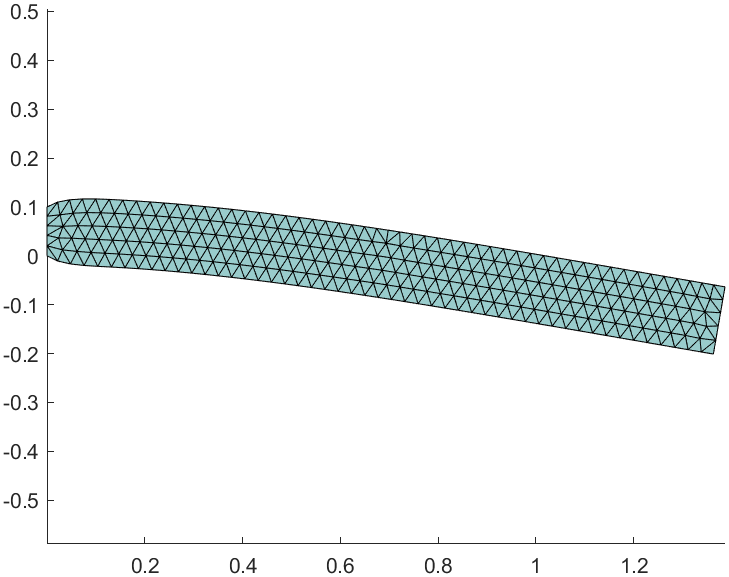}}\hfill
		\subfloat[Deformed, grown configuration after 30 timesteps]{\includegraphics[width=0.3\columnwidth]{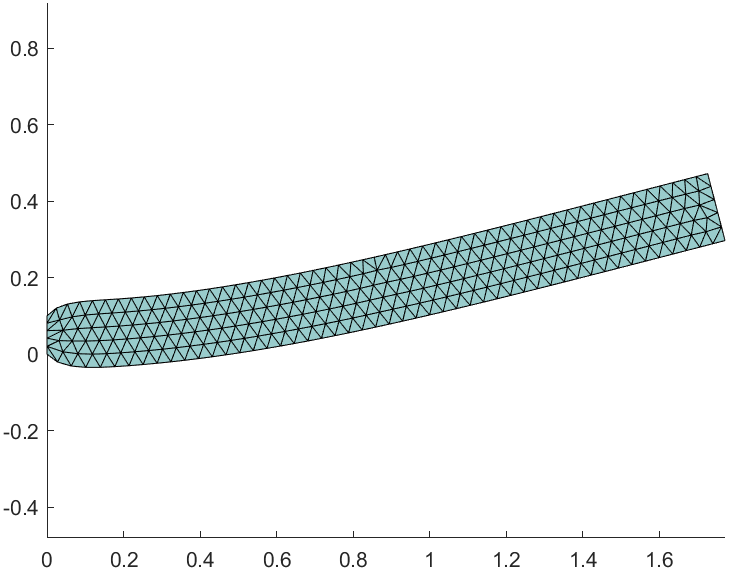}}\\
		\subfloat[Reference configuration]{\includegraphics[width=0.3\columnwidth]{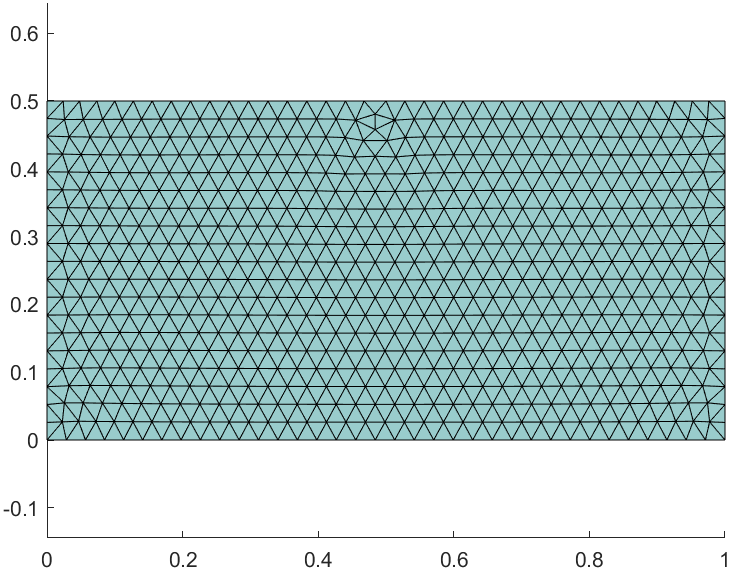}}\hfill
		\subfloat[Grown configuration after 50 timesteps]{\includegraphics[width=0.3\columnwidth]{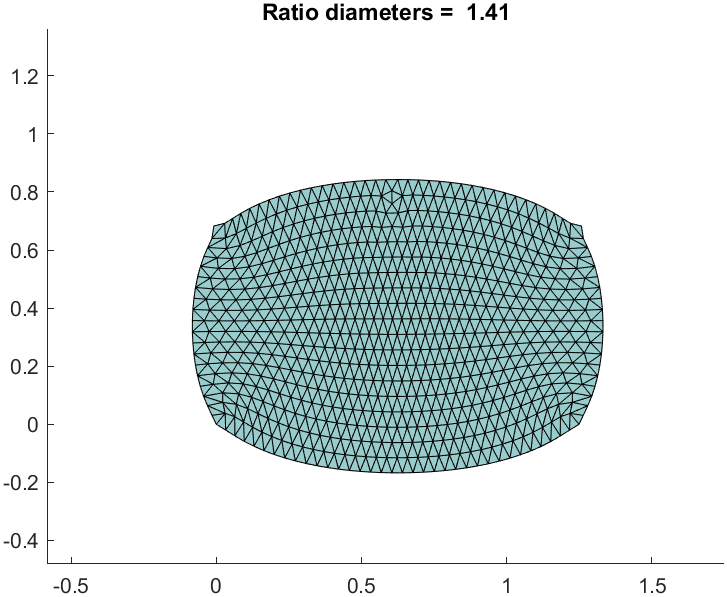}}\hfill
		\subfloat[Grown configuration after 500 timesteps]{\includegraphics[width=0.3\columnwidth]{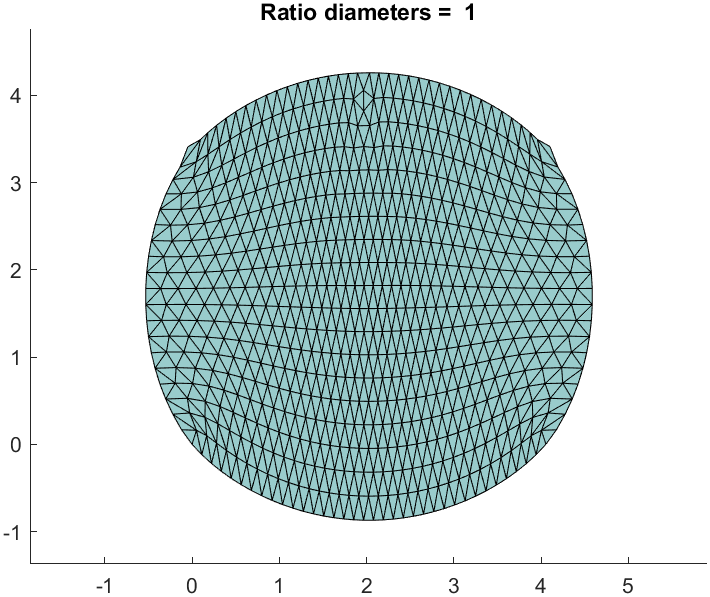}}\\
		\caption{Optimal grown shapes. Doubly-clamped (top), cantilever (center), perimeter (bottom).}
		\label{defd}
	\end{figure}
	In particular, Fig. \ref{defd} shows the deformed configurations for all cases, at an intermediate and at the final step as reported in Tab. \ref{table:1}. For relatively small times the beam configurations are consistent with the purely-elastic response: the external loads overcome the growth effect. Conversely, at later times, the cumulative growth reverse the trend. It is well known that, even in purely elastic settings, stiffness can increase only due to the shape. In the presence of growth, the beams appear to exploit this mechanism, adapting their geometry to increase structural stiffness.
	In the perimeter-driven case, the body evolves toward a circular shape, as expected from the classical isoperimetric principle.

	The fields $(\bfE_g^{(i)})_{11}$, $(\bfE_g^{(i)})_{22}$ are shown in Fig. \ref{g11d}, \ref{g22d}. It is clear that $\bfE_g$ components are higher where material fibers are more elongated in that directions (compare with Fig. \ref{defd}). Conversely, $(\bfE_g^{(i)})_{22}$ is quite constant for the beam cases. It can be appreciated in the perimeter case that $(\bfE_g^{(i)})_{11}$ and $(\bfE_g^{(i)})_{22}$ are ``complementary'': this is due to the fact that the trace of $\bfE_g^{(i)}$ is constant.
	\begin{figure}[!hbt]
		\centering
		\subfloat[$i=15$]{\includegraphics[trim={0 0 0 0.6cm},clip, width=0.45\columnwidth]{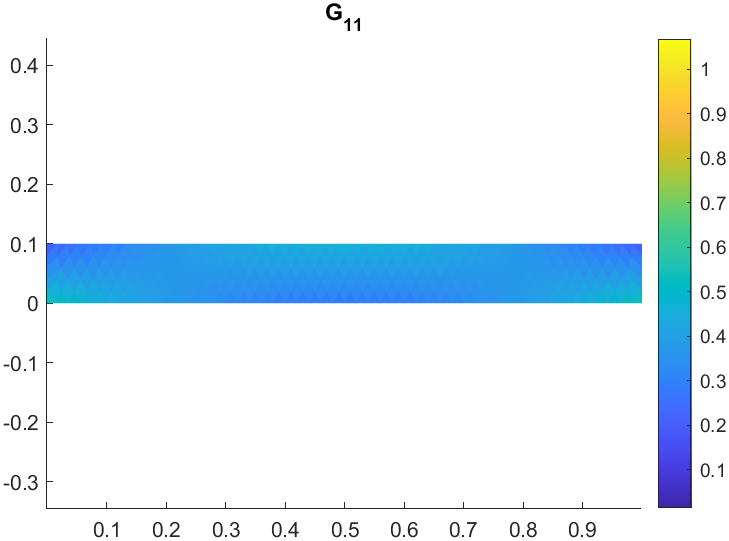}}\hfill
		\subfloat[$i=30$]{\includegraphics[trim={0 0 0 0.6cm},clip,width=0.45\columnwidth]{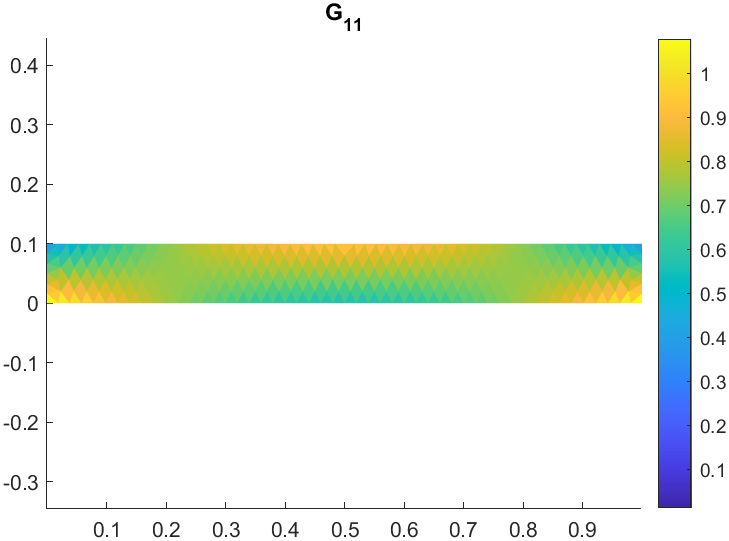}}\\
		\subfloat[$i=15$]{\includegraphics[trim={0 0 0 0.6cm},clip,width=0.45\columnwidth]{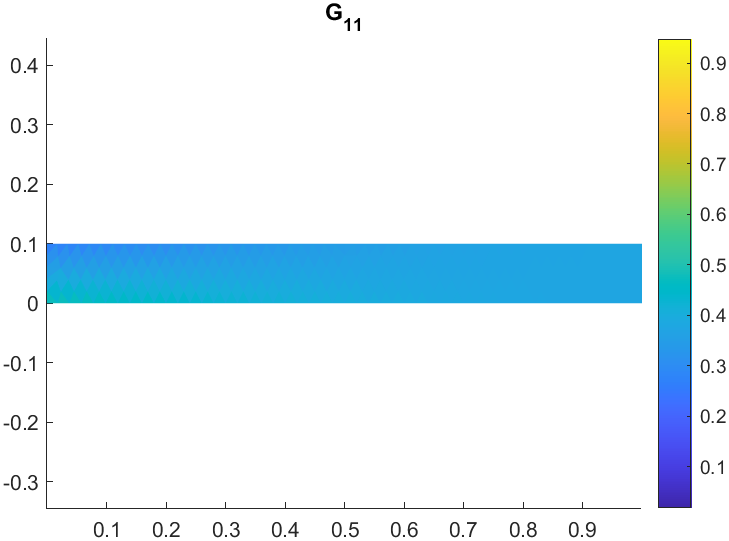}}\hfill
		\subfloat[$i=30$]{\includegraphics[trim={0 0 0 0.6cm},clip,width=0.45\columnwidth]{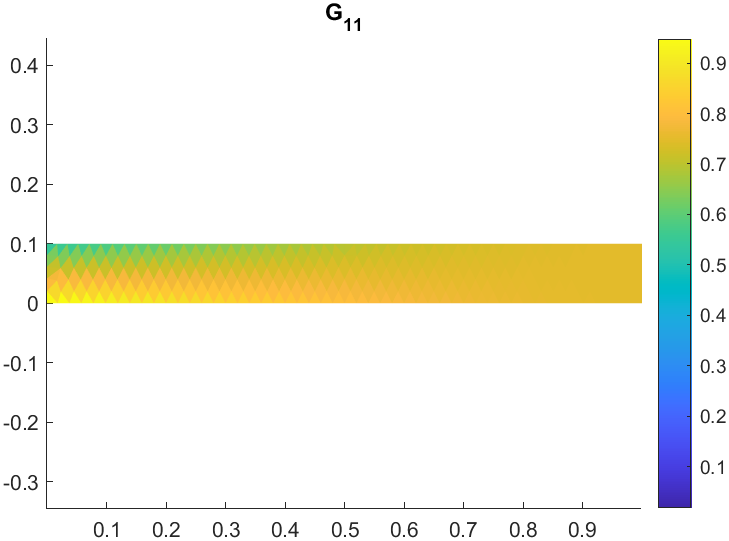}}\\
		\subfloat[$i=50$]{\includegraphics[trim={0 0 0 0.6cm},clip,width=0.45\columnwidth]{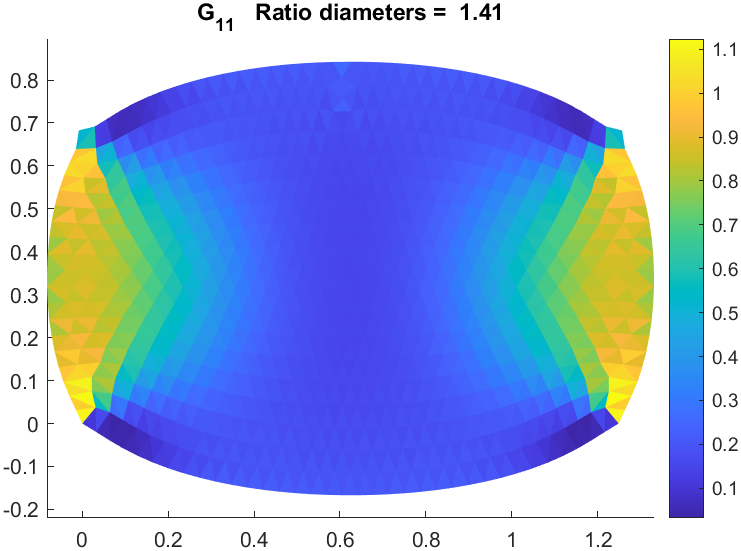}}\hfill
		\subfloat[$i=500$]{\includegraphics[trim={0 0 0 0.6cm},clip,width=0.45\columnwidth]{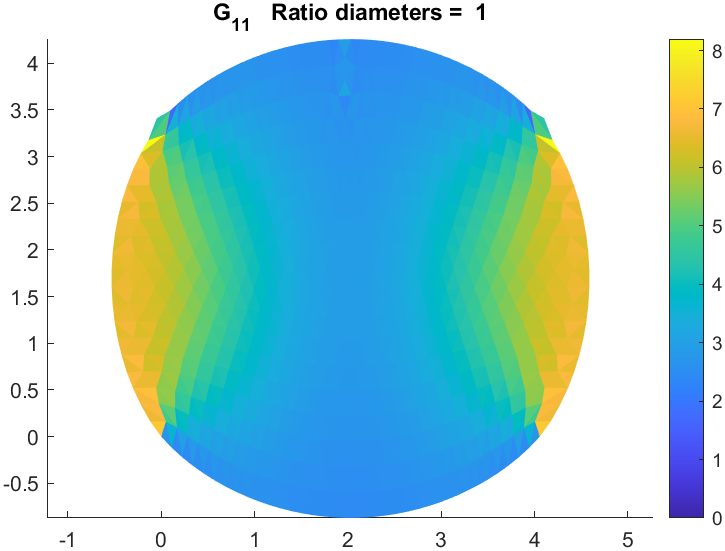}}\\
		\caption{$(\bfE_g^{(i)})_{11}$. Doubly-clamped (top), cantilever (center), perimeter (bottom).}
		\label{g11d}
	\end{figure}

	\begin{figure}[!hbt]
		\centering
		\subfloat[$i=15$]{\includegraphics[trim={0 0 0 0.6cm},clip,width=0.45\columnwidth]{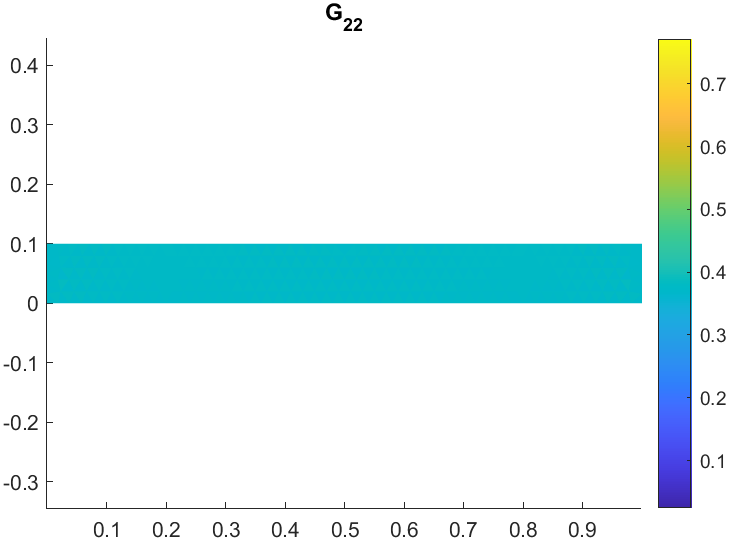}}\hfill
		\subfloat[$i=30$]{\includegraphics[trim={0 0 0 0.6cm},clip,width=0.45\columnwidth]{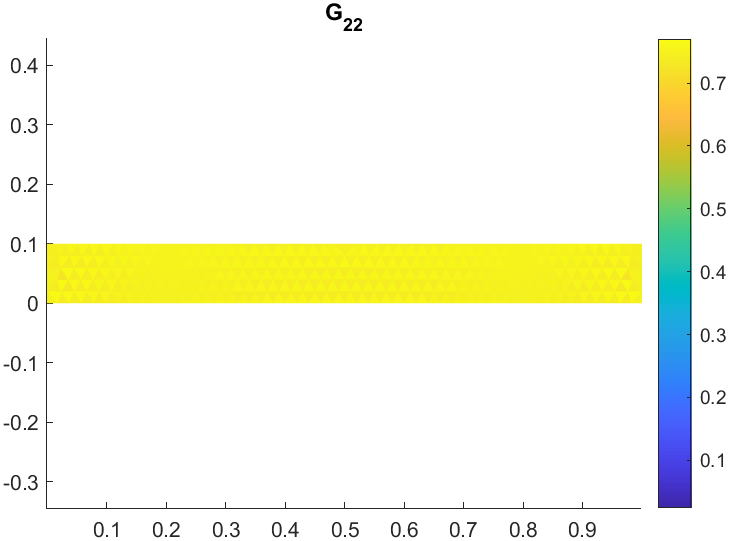}}\\
		\subfloat[$i=15$]{\includegraphics[trim={0 0 0 0.6cm},clip,width=0.45\columnwidth]{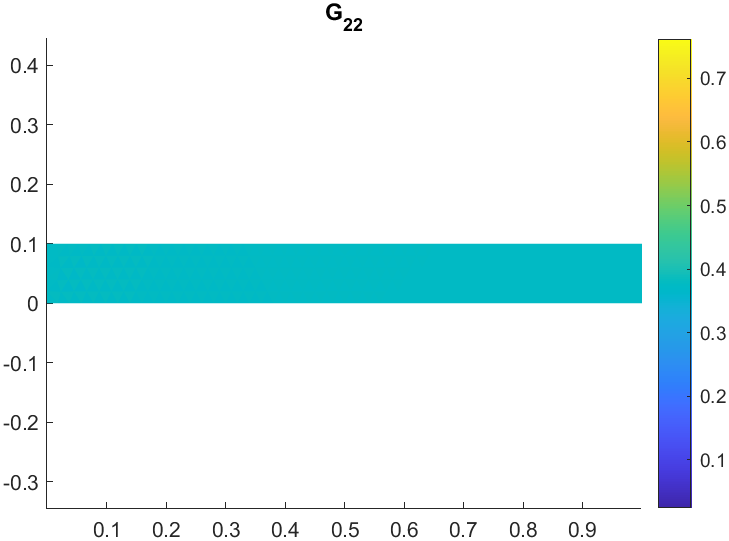}}\hfill
		\subfloat[$i=30$]{\includegraphics[trim={0 0 0 0.6cm},clip,width=0.45\columnwidth]{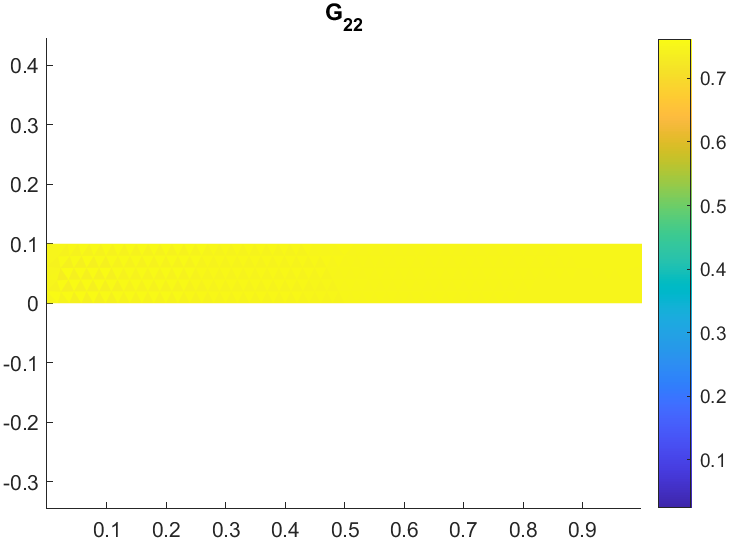}}\\
		\subfloat[$i=50$]{\includegraphics[trim={0 0 0 0.6cm},clip,width=0.45\columnwidth]{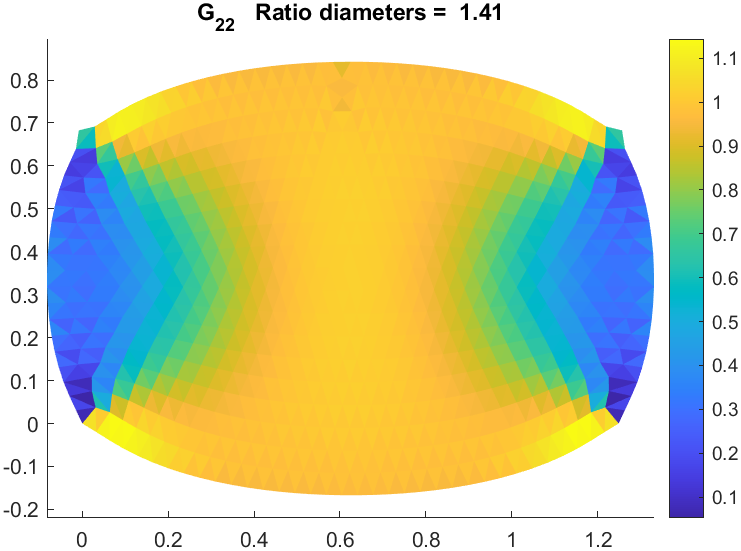}}\hfill
		\subfloat[$i=500$]{\includegraphics[trim={0 0 0 0.6cm},clip,width=0.45\columnwidth]{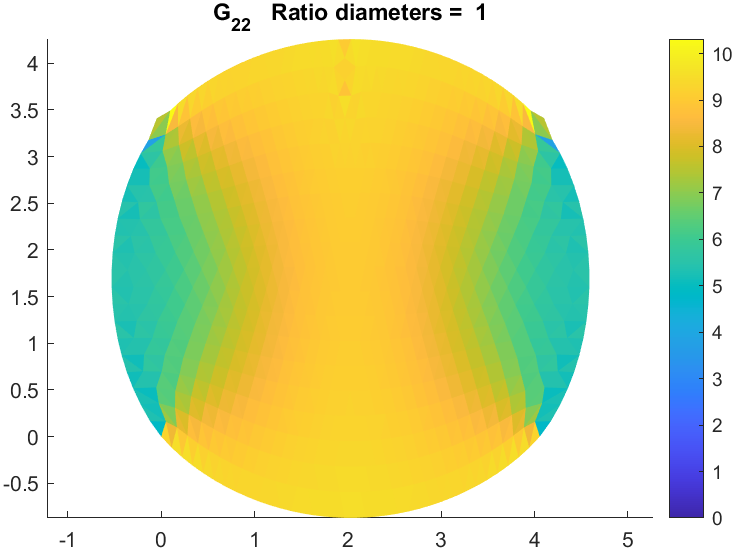}}
		\caption{$(\bfE_g^{(i)})_{22}$ Doubly-clamped (top), cantilever (center), perimeter (bottom).}
		\label{g22d}
	\end{figure}
	The longitudinal Cauchy stress component $(\bfT^{(i)})_{11}$ is shown in Fig. \ref{t11d} for the beam cases only. In the doubly-clamped configuration, the stress is predominantly compressive. In contrast, in the cantilever case, the beam can accommodate growth through elongation (see again Fig. \ref{defd}), thereby reducing the stress magnitude.
	\begin{figure}[!hbt]
		\centering
		\subfloat[$i=15$]{\includegraphics[trim={0 0 0 0.6cm},clip,width=0.45\columnwidth]{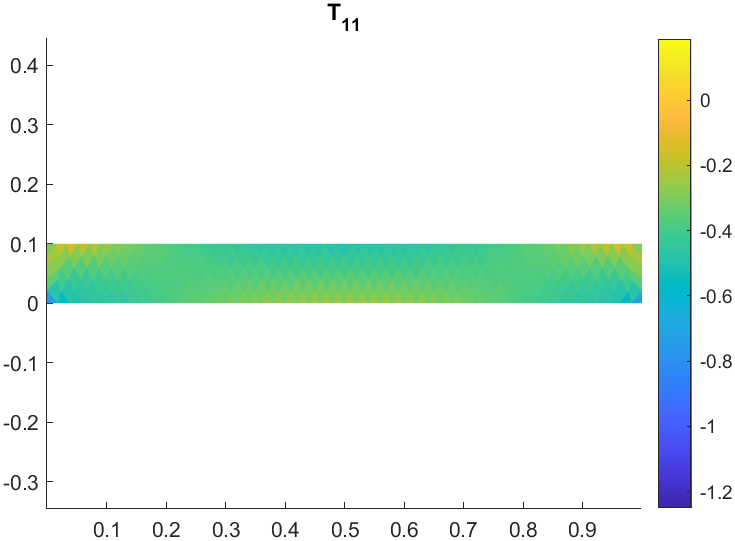}}\hfill
		\subfloat[$i=30$]{\includegraphics[trim={0 0 0 0.6cm},clip,width=0.45\columnwidth]{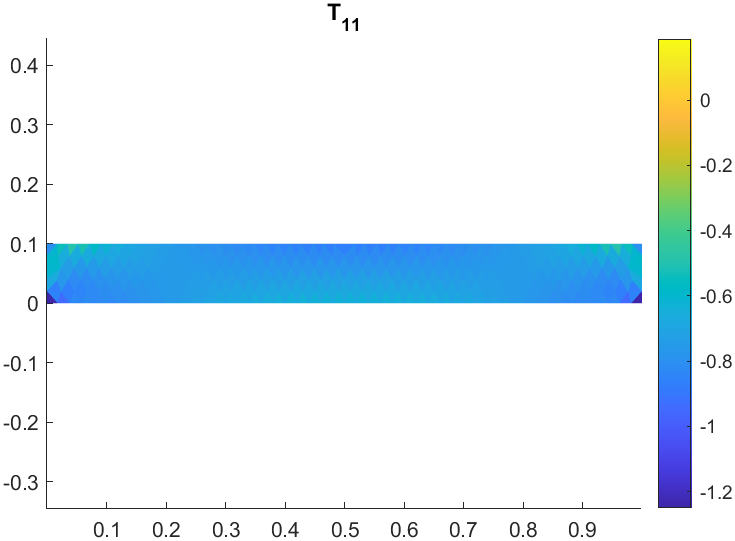}}\\
		\subfloat[$i=15$]{\includegraphics[trim={0 0 0 0.6cm},clip,width=0.45\columnwidth]{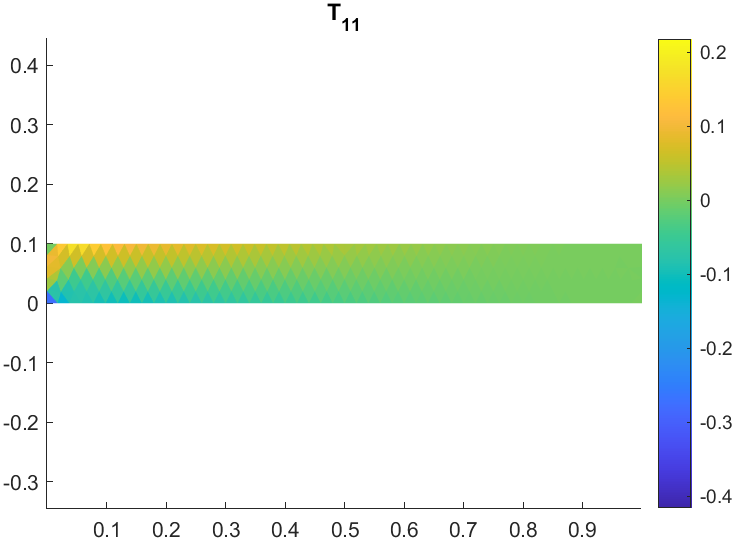}}\hfill
		\subfloat[$i=30$]{\includegraphics[trim={0 0 0 0.6cm},clip,width=0.45\columnwidth]{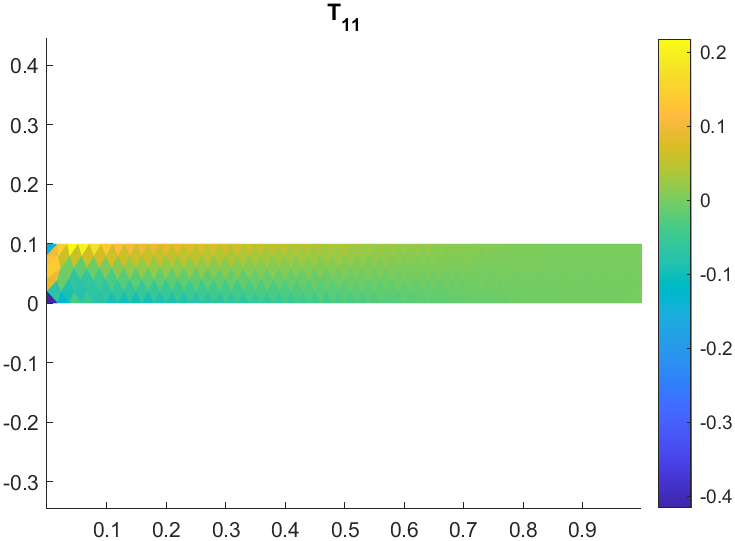}}\\
		\caption{$\bfT^{(i)}_{11}$ Doubly-clamped (top), cantilever (bottom)}
		\label{t11d}
	\end{figure}

	Fig. \ref{t011d} displays the residual stress (longitudinal component), i.e., the stress field when external loads are removed after growth. More precisely, the residual stress tensor is $\bfT_0^{(i)} := \bbC(\bfE(\bfu^{(i)}_0)-\bfE^{(i)}_g)$ where $\bfu_0$ solves
	\[\int_\Omega \bbC [\bfE(\bfu^{(i)}_0) - \bfE_g^{(i)}]\cdot \bfE(\phib) \ d\calL^2 = 0 \qquad \forall \phib \in\calU_0.\]

	Except for localized boundary effects, the residual stress is nearly uniform throughout the beam core.
	\begin{figure}[!hbt]
		\centering
		\subfloat[$i=15$]{\includegraphics[trim={0 0 0 0.6cm},clip,width=0.45\columnwidth]{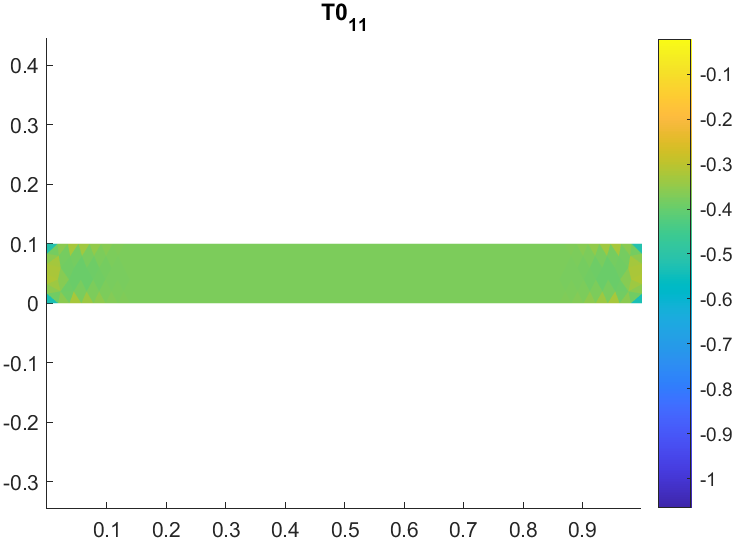}}\hfill
		\subfloat[$i=30$]{\includegraphics[trim={0 0 0 0.6cm},clip,width=0.45\columnwidth]{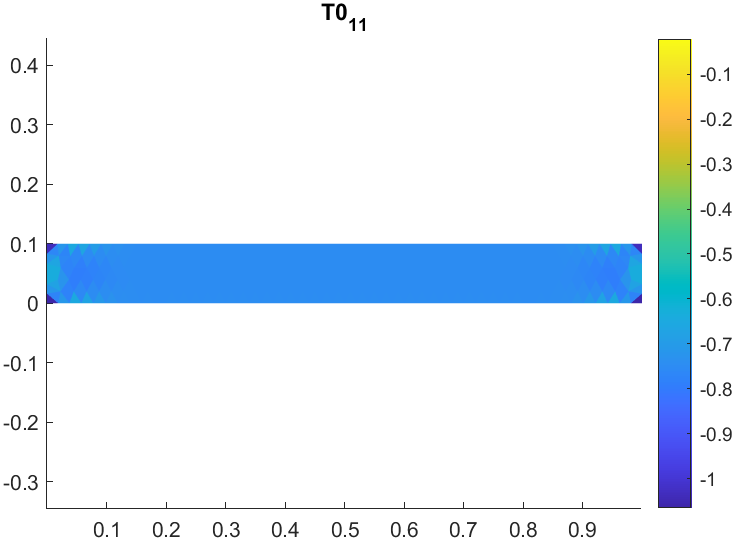}}\\
		\subfloat[$i=15$]{\includegraphics[trim={0 0 0 0.6cm},clip,width=0.45\columnwidth]{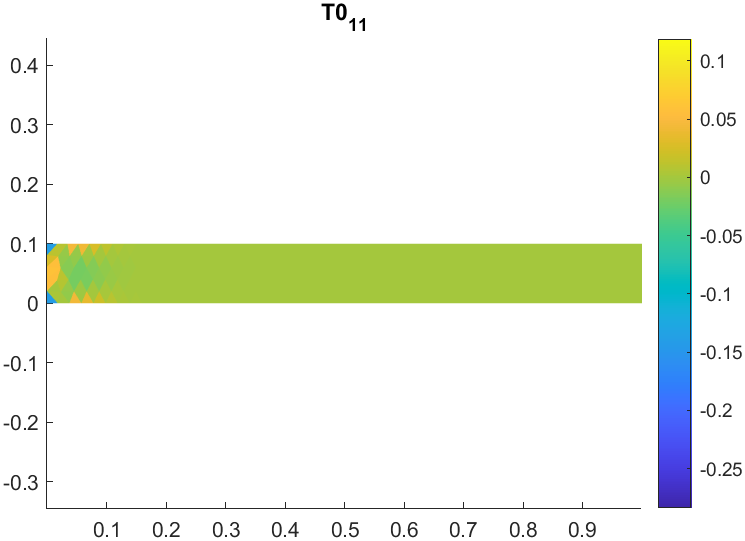}}\hfill
		\subfloat[$i=30$]{\includegraphics[trim={0 0 0 0.6cm},clip,width=0.45\columnwidth]{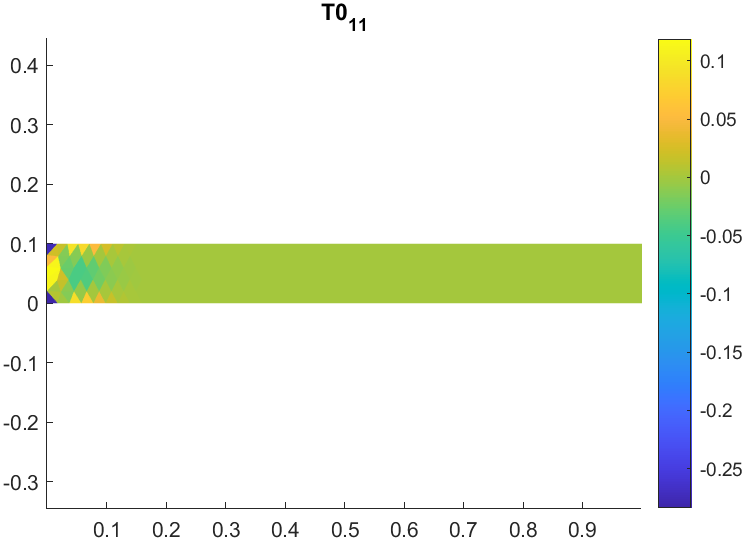}}
		\caption{Residual longitudinal stress. Doubly-clamped (top) and cantilever (bottom)}
		\label{t011d}
	\end{figure}

	For the perimeter case, we report the radial and tangential (hoop) stresses in Fig. \ref{vmd}. Since in this case growth happens without external loads, these can be interpreted as residual stress, as well. It is evident that the material is more stressed in an annular region near the boundary, while the core is definitely less stressed. In particular, the core is compressed while the outer region is under tension. Since the radial stress must vanish on the boundary (no external loads), the tensile stress near the edge is due to the tangential component.

	Interestingly, a similar stress distribution has been observed experimentally and numerically in murine and human tumors \cite{Stylianopoulos2012, Ambrosi2017}: they observed a compressive residual stress in the kernel and a tensile residual stress in the outer shell of the tumor.

	Finally, we remark that each finite element in our discretization may be interpreted as representing a biological cell within a growing tissue.
	\begin{figure}[!hbt]
		\centering
		\subfloat[radial stress $i=500$]{\includegraphics[trim={0 0 0 0.6cm},clip,width=0.45\columnwidth]{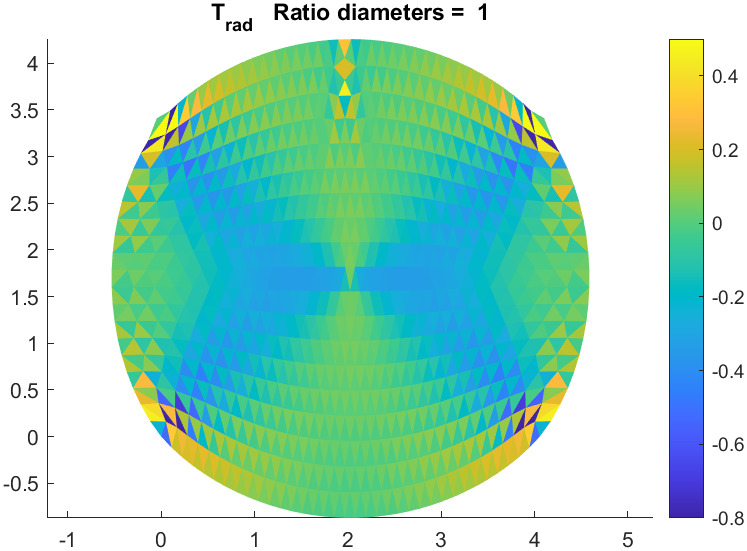}}\hfill
		\subfloat[tangential stress $i=500$]{\includegraphics[trim={0 0 0 0.6cm},clip,width=0.45\columnwidth]{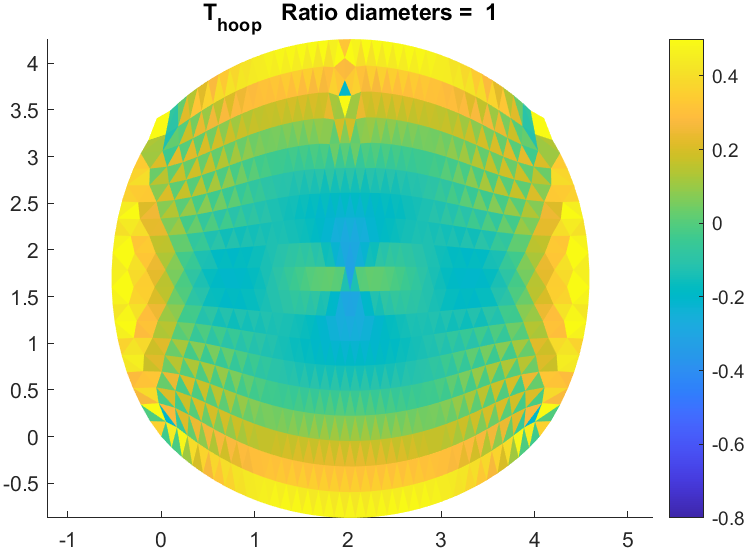}}
		\caption{Residual radial and tangential  stress components, perimeter case}
		\label{vmd}
	\end{figure}

	\section{Approximate analytical solution}\label{secana}

	In this section we derive a closed-form approximation of the discrete growth evolution introduced in Sections \ref{secfem} and \ref{secnum} by temporarily neglecting the nonlinear eigenvalue constraint \eqref{dper5} enforcing accretion in \eqref{workopt22} and \eqref{peropt22}.
	The goal is to make explicit the algebraic structure of the incremental problem and to reveal its interpretation as a projected gradient descent. The analytical solution derived here therefore applies, in principle, in the regime where the monotonicity constraint is either inactive or automatically satisfied. One can verify a posteriori whether the analytical solution satisfies \eqref{dper5}.

	We treat both the perimeter and external work cases in a unified way by writing the objective functional as $\Phi = \Phi(\widehat \bfu^{(i)})$. The discrete minimization problem reads:
	\[
	\begin{cases}
		\displaystyle \min_{\widehat\bfu^{(i)}\in\widehat{V}_0^h\atop \bar\gammab^{(i)}\in\bbR^{3N_e}} & \Phi(\widehat\bfu^{(i)}) + \frac12 \frac{1}{2\tau} \bfL (\bar\gammab^{(i)}-\bar\gammab^{(i-1)})\cdot (\bar\gammab^{(i)}-\bar\gammab^{(i-1)})\\
		& \displaystyle \bfK\, \widehat\bfu^{(i)}=\bfB\,\bar\gammab^{(i)} + \bff,\\
		& \displaystyle \bfa\cdot(\bar\gammab^{(i)}-\bar\gammab^{(i-1)})  = \Gamma|\Omega| \quad \text{ or } \quad \bfA(\bar\gammab^{(i)}-\bar\gammab^{(i-1)}) = \Gamma \widetilde{\mathbf{1}}
	\end{cases}
	\]
	where $\widetilde{\mathbf{1}} = \sum_{e=1}^{N_e}\widetilde{\bfe}_e$.
	Solving explicitly the equilibrium equation gives  $\widehat\bfu^{(i)} = \bfK\inv(\bfB\,\bar\gammab^{(i)} + \bff)$. Substituting into the objective functional reduces the problem to a constrained minimization in the growth variable alone:
	\[
	\begin{cases}
		\displaystyle \min_{\bar\gammab^{(i)}\in\bbR^{3N_e}} &  \Psi(\bar\gammab^{(i)})  + \frac12 \frac{1}{2\tau} \bfL (\bar\gammab^{(i)}-\bar\gammab^{(i-1)})\cdot (\bar\gammab^{(i)}-\bar\gammab^{(i-1)})\\
		& \displaystyle \bfa\cdot(\bar\gammab^{(i)}-\bar\gammab^{(i-1)})  = \Gamma|\Omega| \quad \text{ or } \quad \bfA(\bar\gammab^{(i)}-\bar\gammab^{(i-1)}) = \Gamma \widetilde{\mathbf{1}}.
	\end{cases}
	\]
	where for brevity $\Psi(\bar\gammab^{(i)}):=\Phi(\bfK\inv(\bfB\,\bar\gammab^{(i)} + \bff))$.
	To derive first-order optimality conditions, we introduce Lagrange multipliers.
	The Lagrangian for the minimization problem is
	\[\calL(\bar\gammab^{(i)}, \lambda^{(i)}/\pmb{\lambda}^{(i)}) = \Psi(\bar\gammab^{(i)})  + \frac12 \frac{1}{2\tau} \bfL (\bar\gammab^{(i)}-\bar\gammab^{(i-1)})\cdot (\bar\gammab^{(i)}-\bar\gammab^{(i-1)}) + \begin{cases}
		\lambda^{(i)}(\bfa\cdot(\bar\gammab^{(i)}-\bar\gammab^{(i-1)})  - \Gamma|\Omega|)\\
		\pmb{\lambda}^{(i)}\cdot(\bfA(\bar\gammab^{(i)}-\bar\gammab^{(i-1)}) - \Gamma \widetilde{\mathbf{1}})
	\end{cases}\]
	where $\lambda^{(i)}\in\bbR$ and $\pmb{\lambda}^{(i)}\in\bbR^{N_e}$ are the Lagrange multiplier for the global and local mass balance, respectively.
	The Karush–Kuhn–Tucker (KKT) conditions (stationarity of the Lagrangian with respect to  the Lagrange multiplier and $\bar\gammab^{(i)}$) reads:
	\[\begin{cases}
		\bfa\cdot(\bar\gammab^{(i)}-\bar\gammab^{(i-1)})  - \Gamma|\Omega| = 0,\\
		\nabla\Psi + \frac{1}{2\tau}\bfL (\bar\gammab^{(i)} - \bar\gammab^{(i-1)}) + \lambda^{(i)}\bfa = \mathbf{0},
	\end{cases} \qquad \begin{cases}
		\bfA(\bar\gammab^{(i)}-\bar\gammab^{(i-1)})  - \Gamma\widetilde{\mathbf{1}} = \mathbf{0},\\
		\nabla\Psi + \frac{1}{2\tau}\bfL (\bar\gammab^{(i)} - \bar\gammab^{(i-1)}) + \bfA\tran\pmb{\lambda}^{(i)} = \mathbf{0},
	\end{cases}\]
	depending on whether the mass balance is in global or local form.
	Since $\bfL$ (defined in \eqref{eqL}) is symmetric and strictly positive definite, it is invertible.  Solving for $(\bar\gammab^{(i)} - \bar\gammab^{(i-1)})$ in the second equation and substituting in the first to obtain $\lambda^{(i)}$, after simple manipulations we arrive for the global mass balance case to
	\begin{equation}\label{cfs}
		\begin{cases}
			\displaystyle -\lambda^{(i)} = \frac{\bfa\cdot(\bfL\inv\nabla\Psi )  + \frac{1}{2\tau}\Gamma|\Omega|}{|\sqrt{\bfL\inv}\bfa|^2},\\
			\displaystyle\bar\gammab^{(i)} - \bar\gammab^{(i-1)} = -2\tau \sqrt{\bfL\inv}(\bfI - \frac{\sqrt{\bfL\inv}\bfa}{|\sqrt{\bfL\inv}\bfa|}\otimes \frac{\sqrt{\bfL\inv}\bfa}{|\sqrt{\bfL\inv}\bfa|}) \sqrt{\bfL\inv}\nabla\Psi   + \Gamma|\Omega|\frac{\bfL\inv\bfa}{|\sqrt{\bfL\inv}\bfa|^2}.
	\end{cases}\end{equation}
The matrix $\sqrt{\bfL^{-1}}$ such that $\sqrt{\bfL^{-1}}\sqrt{\bfL^{-1}} = \bfL^{-1}$ is well-defined and is given by $$\sqrt{\bfL^{-1}} = \sum_{e=1}^{N_e}\frac{1}{\sqrt{2|T_e|}}(\bar\bfe_{3e-2}\otimes\bar\bfe_{3e-2}+ \bar\bfe_{3e-1}\otimes\bar\bfe_{3e-1}+ \bar\bfe_{3e}\otimes\bar\bfe_{3e}).$$
	For the local mass balance case, with similar computations and crucially observing that $\bfA\bfL\inv\bfA\tran$ is invertible\footnote{For a matrix $\bfM\in \bbR^{m\times n}$ with $m\leq n$ with full rank it is well known that $\bfM\bfM\tran$ is invertible. Apply to $\bfA\sqrt{\bfL\inv}$.}, we obtain
	\begin{equation}\label{cfsloc}
		\begin{cases}
			\displaystyle  -\pmb{\lambda}^{(i)} = \bfU\sqrt{\bfL\inv}\nabla\Psi  + \frac{1}{2\tau}\Gamma(\bfA\bfL\inv\bfA\tran)\inv\widetilde{\mathbf{1}}, \\
			\displaystyle\bar\gammab^{(i)} - \bar\gammab^{(i-1)} = -2\tau \sqrt{\bfL\inv}(\bfI - \bfP) \sqrt{\bfL\inv}\nabla\Psi   + \Gamma\sqrt{\bfL\inv}\bfV\widetilde{\mathbf{1}},
		\end{cases}
	\end{equation}
where $$\bfP:=\sqrt{\bfL\inv}\bfA\tran(\bfA\bfL\inv\bfA\tran)\inv\bfA\sqrt{\bfL\inv}, \qquad \bfV := \sqrt{\bfL\inv}\bfA\tran(\bfA\bfL\inv\bfA\tran)\inv, \qquad \bfU:=(\bfA\bfL\inv\bfA\tran)\inv\bfA\sqrt{\bfL\inv}.$$
It is easy to see that $\bfP$ in \eqref{cfsloc} is an orthogonal projection, since $\bfP^2=\bfP$ and $\bfP = \bfP\tran$. Moreover
 $\bfA\sqrt{\bfL\inv} \bfV = \bfI$ and  $\bfU \sqrt{\bfL\inv}\bfA\tran  = \bfI$, i.e., $\bfV$ is a right inverse for $\bfA\sqrt{\bfL\inv}$ and $\bfU$ is a left inverse for $(\bfA\sqrt{\bfL\inv})\tran$.

	Some general remarks are in order.
	It is remarkable that $\lambda^{(i)}$, $\pmb{\lambda}^{(i)}$ and $\bar\gammab^{(i)} - \bar\gammab^{(i-1)}$ depend on time only through $\nabla\Psi=\nabla\Psi(\bar\gammab^{(i)})$. Thus, if $\nabla\Psi$ does not depend on time, the growth proceeds linearly in time.
	The mass increment decomposes into two contributions:
	a mechanics-driven term proportional to $\nabla\Psi$ and
	a purely volumetric term proportional to $\Gamma$.
	The projectors $(\bfI - \frac{\sqrt{\bfL\inv}\bfa}{|\sqrt{\bfL\inv}\bfa|}\otimes \frac{\sqrt{\bfL\inv}\bfa}{|\sqrt{\bfL\inv}\bfa|})$ and $(\bfI - \bfP)$
	remove the volumetric component of the mechanically induced term in order to satisfy the mass balance.
	In particular, growth evolves along a projected gradient descent direction.

	Note also that the solution in \eqref{cfs} and \eqref{cfsloc} satisfies the eigenvalue constraint \eqref{dper5} at every iteration if the term $\frac{\nabla\Psi}{\frac{1}{2\tau}}$ remains sufficiently small. Indeed, the volumetric term is a vector of nonnegative components, while  the mechanics-driven term has no predefined sign\footnote{Actually, the matrix $\sqrt{\bfL\inv}(\bfI - \bfP) \sqrt{\bfL\inv}$ is positive semidefinite. It follows from the fact that a projection ($\bfI - \bfP$) is always positive semidefinite and that $\sqrt{\bfL\inv}$ is positive definite}. If it is the case that $\frac{\nabla\Psi}{\frac{1}{2\tau}}$ remains sufficiently small during growth, the eigenvalue constraint is satisfied at all times, and a solution to \eqref{cfs} or \eqref{cfsloc} is (at least) an admissible solution in the exact sense. At the limit case $\frac{1}{2\tau} \to \infty$ the mechanically induced contribution vanishes.

	We specialize now the analytical solution when the objective function is external work, comparing then the numerical results previously obtained in Sec. \ref{secnum} and the analytical solution.
	In this case $\Phi(\widehat\bfu^{(i)}) = \bff \cdot  \widehat\bfu^{(i)}$ and so $\nabla\Psi = (\bfK\inv\bfB)\tran\bff$.
	The relative Karush–Kuhn–Tucker (KKT) conditions read:
	\[\begin{cases}
		\bfa\cdot(\bar\gammab^{(i)}-\bar\gammab^{(i-1)})  - \Gamma|\Omega| = 0,\\
		(\bfK\inv\bfB)\tran\bff + \frac{1}{2\tau}\bfL (\bar\gammab^{(i)} - \bar\gammab^{(i-1)}) + \lambda^{(i)}\bfa = \mathbf{0},
	\end{cases}\]
	from which
	\begin{equation}\label{cfsw}
		\begin{cases}
			\displaystyle  -\lambda^{(i)} = \frac{\bfa\cdot(\bfL\inv\bfB\tran\bfK\inv\bff )  + \frac{1}{2\tau}\Gamma|\Omega|}{|\sqrt{\bfL\inv}\bfa|^2}, \\
			\displaystyle\bar\gammab^{(i)} - \bar\gammab^{(i-1)} = -2\tau \sqrt{\bfL\inv}(\bfI - \frac{\sqrt{\bfL\inv}\bfa}{|\sqrt{\bfL\inv}\bfa|}\otimes \frac{\sqrt{\bfL\inv}\bfa}{|\sqrt{\bfL\inv}\bfa|}) \sqrt{\bfL\inv}\bfB\tran\bfK\inv\bff   + \Gamma|\Omega|\frac{\bfL\inv\bfa}{|\sqrt{\bfL\inv}\bfa|^2}.
	\end{cases}\end{equation}
	As anticipated in the general case, $\lambda^{(i)}$ and $\bar\gammab^{(i)} - \bar\gammab^{(i-1)}$ do not depend on time since $\nabla\Psi$ does not. Accordingly, growth evolves linearly in time.
	Note also that the solution in \eqref{cfsw} satisfies the  eigenvalue constraint \eqref{dper5} at every iteration if $-2\tau (\bfI - \frac{\sqrt{\bfL\inv}\bfa}{|\sqrt{\bfL\inv}\bfa|}\otimes \frac{\sqrt{\bfL\inv}\bfa}{|\sqrt{\bfL\inv}\bfa|}) \sqrt{\bfL\inv}\bfB\tran\bfK\inv\bff   + \Gamma|\Omega|\frac{\bfL\inv\bfa}{|\sqrt{\bfL\inv}\bfa|^2}$, which is time independent, does. Actually, the second term satisfies the constraint, while the first term need not. We can however choose $\frac{1}{2\tau}$  sufficiently large such that the eigenvalue constraint is satisfied, and a solution to \eqref{cfsw} is (at least) an admissible solution of problem \eqref{workopt22}.
	It happens that for the values in Tab. \ref{table:1} the eigenvalue constraint is satisfied (numerics shows that, for the values in Tab. \ref{table:1}, the smallest $\frac{1}{2\tau}$ such that the eigenvalue constraint is satisfied is about $4.6$).

	Figure \ref{ddana} shows the deformed, grown configuration and the field $(\bfE_g^{(i)})_{11}$ after $30$ iterations for the doubly-clamped case. These pictures should be compared to Figs. \ref{defd} and \ref{g11d}.
	After $30$ iterations, the objective function value is $0.00092$ for the analytic solution and $0.00091$ for the complete numerical solution. Moreover, the $L^2(\Omega)$-norm of the difference of the two minimizers is $0.0152$. The two minimizers are therefore practically indistinguishable, confirming the validity of the closed-form solution in this regime.
	\begin{figure}[!hbt]
		\centering
		\subfloat[Deformed, grown configuration $i=30$]{\includegraphics[trim={0 0 0 0.6cm},clip,width=0.45\columnwidth]{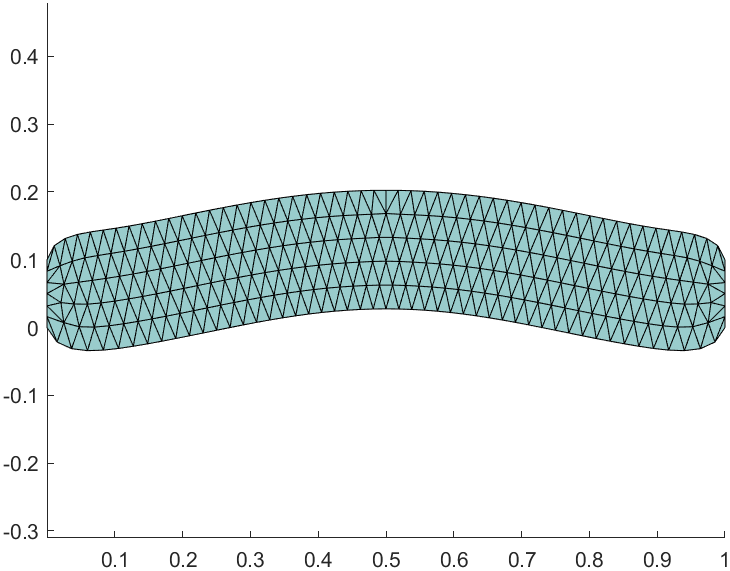}}\hfill
		\subfloat[$(\bfE_g^{(30)})_{11}$]{\includegraphics[trim={0 0 0 0.6cm},clip,width=0.45\columnwidth]{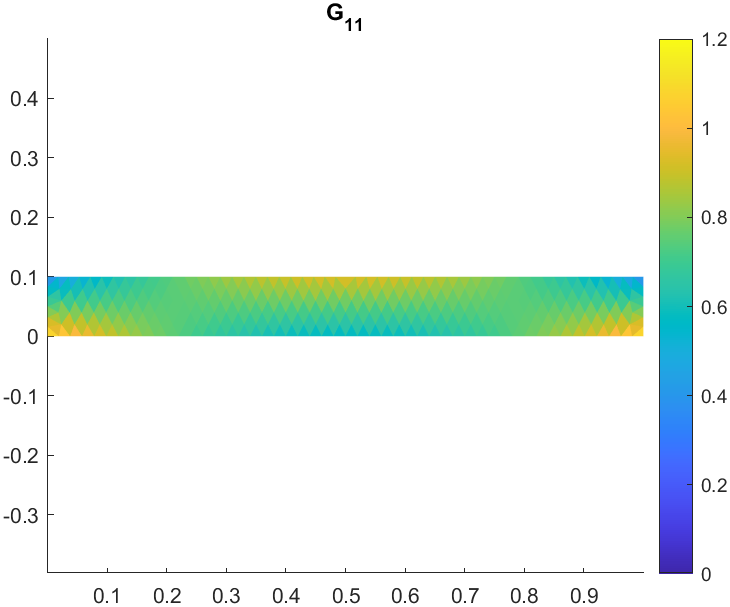}}
		\caption{Analytical solution for the doubly clamped beam growth}
		\label{ddana}
	\end{figure}
	For the perimeter case, Fig. \ref{vmdana} shows the radial and hoop stresses after solving \eqref{cfsloc} after $500$ steps, with the same numerical values of Tab. \ref{table:1}. Note that the second equation in \eqref{cfsloc} is an implicit one in $\bar\gammab^{(i)}$. To simplify, we  evaluate $\nabla\Psi$ in $\bar\gammab^{(i-1)}$. Fig. \ref{vmdana} should be compared with Fig. \ref{vmd}. We remark that with the choice $\frac{1}{2\tau}=100$ the eigenvalue constraint is not satisfied within the first fifty (roughly) iterations. Nevertheless, the final solution is consistent with the numerical simulations. The optimal objective function value is $16.37$ in the complete numerical simulation and $16.31$ for the analytic solution. The difference is probably due to the fact that the analytical solution does not satisfy the eigenvalue constraint at all time.

		\begin{figure}[!hbt]
		\centering
		\subfloat[radial stress $i=500$]{\includegraphics[trim={0 0 0 0.6cm},clip,width=0.45\columnwidth]{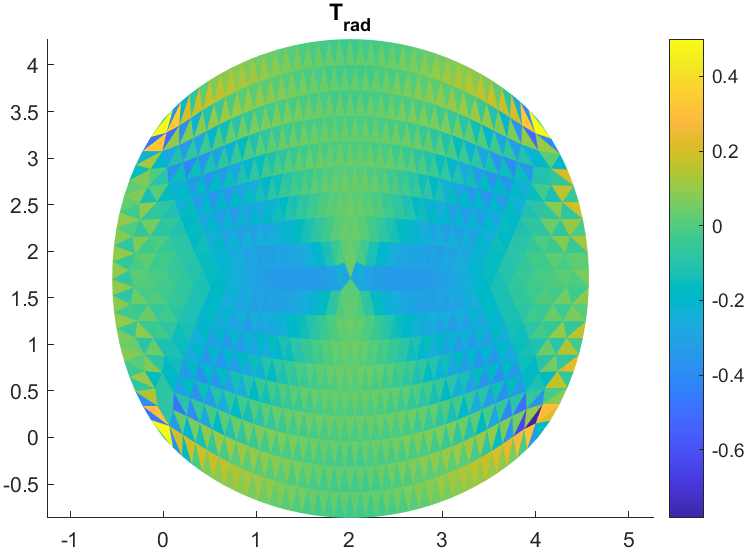}}\hfill
		\subfloat[tangential stress $i=500$]{\includegraphics[trim={0 0 0 0.6cm},clip,width=0.45\columnwidth]{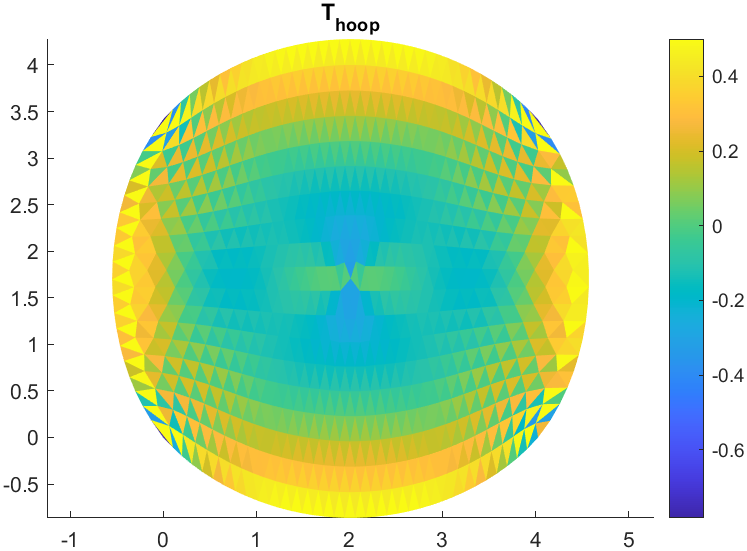}}
		\caption{Residual radial and tangential stress components, perimeter case, analytical solution}
		\label{vmdana}
	\end{figure}

	\section{Conclusions}\label{secconc}

	We have shown, using linearized elasticity as a convenient framework, that growth can be formulated as an optimization-driven process. The central modeling choice is that the growth tensor $\bfE_g$
	is not prescribed phenomenologically through  evolution laws. Instead, its evolution is determined implicitly by the solution of a constrained minimization problem at each incremental step.

	This perspective marks a fundamental distinction with respect to the existing literature: growth is not treated as a descriptive kinematic input, but as the outcome of a mechanical selection principle. The material adapts by choosing, among all admissible growth increments satisfying the imposed constraints, the one that optimizes a given objective functional.

	In this paper, we have considered two different driving mechanisms: minimization of external work (mechanically driven adaptation) or perimeter reduction (geometrically driven evolution). In all cases, our model predicts how growth distributes spatially as a response to the chosen objective.
	In the former case, the body adapts its geometry to increase stiffness and reduce external work. In the latter case, the evolution toward circular shapes could be expected from the isoperimetric principle. These outcomes are predictions of the model, not assumptions embedded into it.

	The linear theory thus serves as a proof of concept: mechanically constrained optimization alone is sufficient to generate growth-induced morphologies. This paves the way to extend our approach to fully nonlinear settings and to more realistic biological or engineered systems, where growth is governed by competing energetic or functional objectives.

	Future developments may include the incorporation of   transport or biochemical processes, and a rigorous analysis of the continuous-time limit of the discrete scheme. From the computational standpoint, more advanced optimization strategies and large-scale simulations could further explore the rich morphological behaviors emerging from growth driven by optimality principles.

\medskip

\noindent
{\bf Acknowledgments.} We gratefully acknowledge the support of the MIT-UNIPI Seed Fund.\\
  R.P.  and  M.P.S. thankfully acknowledges the support of the Italian National Group of Mathematical Physics (GNFM-INdAM).

\end{document}